\documentclass[11pt]{article}
\usepackage{epsfig}
\usepackage{amssymb,amsmath,amsthm,url}
\usepackage{multirow}
\usepackage{booktabs}
\usepackage{natbib}
\usepackage{setspace}

\usepackage{times,amsfonts,eucal}
\usepackage{algorithm}
\usepackage{algorithmic}
\usepackage{graphicx,ifthen}
\usepackage{wrapfig}
\usepackage{latexsym}
\usepackage{color}
\usepackage{mathrsfs}
\usepackage{bm}
\usepackage{bbm}
\usepackage{srctex}

\usepackage{bm}
\usepackage{natbib}
\newtheorem{theorem}{Theorem}[section]

\newtheorem{cor}{Corollary}[section]
\newtheorem{definition}{Definition}

\newtheorem{assu}{Assumption}[section]
\numberwithin{equation}{section}
\numberwithin{theorem}{section}
\numberwithin{lemma}{section}
\numberwithin{pro}{section}
\numberwithin{cor}{section}
\numberwithin{definition}{section}
\numberwithin{cons}{section}
\numberwithin{rem}{section}
\numberwithin{exa}{section}
\numberwithin{table}{section}
\numberwithin{figure}{section}

\newcommand{\vare}{\varepsilon}

\newcommand{\by}{\mathbf{y}}
\newcommand{\bs}{\mathbf{s}}
\newcommand{\bd}{\mathbf{d}}
\newcommand{\bA}{\mathbf{A}}
\newcommand{\bB}{\mathbf{B}}
\newcommand{\bH}{\mathbf{H}}
\newcommand{\bJ}{\mathbf{J}}
\newcommand{\bI}{\mathbf{I}}
\newcommand{\bN}{\mathbf{N}}
\newcommand{\bQ}{\mathbf{Q}}
\newcommand{\bR}{\mathbf{R}}
\newcommand{\bV}{\mathbf{V}}
\newcommand{\bzero}{\mathbf{0}}
\newcommand{\la}{\langle}
\newcommand{\ra}{\rangle}
\newcommand{\varp}{\varphi}
\newcommand{\btheta}{\mathbf{\theta}}
\newcommand{\bTheta}{\mathbf{\Theta}}
\newcommand{\ubA}{\underline{\bA}}
\newcommand{\ubB}{\underline{\bB}}
\newcommand{\ubd}{\underline{\bd}}
\newcommand{\utheta}{\underline{\theta}}
\newcommand{\ubtheta}{\underline{\btheta}}

\newcommand{\T}{\top}
\newcommand\cC{{\mathcal C}}
\newcommand{\cN}{{\mathcal N}}

\newcommand{\intt}{\int\hspace{-.2cm}\int}
\newcommand{\inttt}{\int\hspace{-.2cm}\int\hspace{-.2cm}\int }

\def\beq{\begin{equation}}
\def\eeq{\end{equation}}
\def\bals{\begin{align*}}
\def\eals{\end{align*}}

\def\bal{\begin{align}}
\def\eal{\end{align}}

\allowdisplaybreaks

\begin{document}

\title{Functional generalized autoregressive conditional heteroskedasticity\footnote{This research was partially supported by NSF grants DMS 1209226, DMS 1305858 and DMS 1407530
}}

\author{
Alexander Aue\footnote{Department of Statistics, University of California, Davis, CA 95616, USA, email: \tt{aaue@ucdavis.edu}}
\and Lajos Horv\'ath\footnote{Department of Mathematics, University of Utah, Salt Lake City, UT 84112, USA, emails: \tt{horvath@math.utah.edu}, {\tt danielpellatt@gmail.com}}
\and Daniel F.\ Pellatt$^\ddagger$
}
\date{\today}
\maketitle

\begin{abstract}
\setlength{\baselineskip}{1.8em}
Heteroskedasticity is a common feature of financial time series and is commonly addressed in the model building process through the use of ARCH and GARCH processes. More recently multivariate variants of these processes have been in the focus of research with attention given to methods seeking an efficient and economic estimation of a large number of model parameters. Due to the need for estimation of many parameters, however, these models may not be suitable for modeling now prevalent high-frequency volatility data. One potentially useful way to bypass these issues is to take a functional approach. In this paper, theory is developed for a new functional version of the generalized autoregressive conditionally heteroskedastic process, termed fGARCH. The main results are concerned with the structure of the fGARCH(1,1) process, providing criteria for the existence of strictly stationary solutions both in the space of square-integrable and continuous functions. An estimation procedure is introduced and its consistency and asymptotic normality verified. A small empirical study highlights potential applications to intraday volatility estimation. 
\medskip \\
\noindent {\bf Keywords:} Econometrics; Financial time series; Functional data; GARCH processes; Stationary solutions

\noindent {\bf MSC 2010:} Primary: 62M10, Secondary: 62P20, 91B84
\end{abstract}

\setlength{\baselineskip}{1.8em}

%%%%%%%%%%%%%%%%%%%%%%%%%%%%%%%%%%%%%%%%%%%%%%%%%%%

\section{Introduction}
\label{sec:intro}

Modeling volatility is one of the prime objectives of financial time series analysis. Research in the area has increased rapidly since the seminal contributions by Engle (1982), who introduced the autoregressive conditional heteroskedastic, ARCH, model, and Bollerslev (1986), who introduced the generalized ARCH, GARCH, model to deal with non-constant and randomly changing volatilities for univariate time series. Methodology based on these processes and their modifications have become major theoretical and applied tools to analyze stock returns, exchange rates and more. The book Francq and Zako\"ian (2010) provides an excellent introduction to volatility processes and their applications. The interested reader may find detailed accounts on theoretical properties such as conditions for the existence of stationary solutions and finiteness of moments as well as methodology for the estimation of model parameters. 

With ever more data being collected in today's financial sector, there has come a need to address new phenomena observed when dealing with high-frequency, intra-day observations. While multivariate volatility processes such as the ones of Engle (2002) and Engle and Kroner (1995) might be used to model high-frequency returns, these models are typically associated with the difficult task of estimating a large number of parameters (see Chapter 11 of Francq and Zako\"ian, 2010). Surveys of multivariate GARCH models and related volatility processes may be found in Bauwens et al.\ (2006) and Silvennoinen and Ter\"{a}svirta (2009). To alleviate the numerical burden of the parameter estimation, the variance targeting method of Engle and Mezrich (1996) can be used, which was recently extended to cover multivariate GARCH models by Francq et al.\ (2011, 2015). These methods, however, may still fail to produce satisfactory results in a high-frequency setting.

Modeling data as a collection of functions was popularized through the work of Ramsay and Silvermann (2005). While this approach has now started to become relevant for the analysis of high-frequency volatility data, much of the research in the area has so far been devoted to homoskedastic functional time series. The books by Bosq (2010) and Horv\'ath and Kokoszka (2012) provide an account of the state-of-the-art research on dependent sequences of random functions, including functional autoregressive processes. More recently, prediction of functional time series was considered in Aue et al.\ (2015), while estimation in dependent functional linear models was discussed in Aue et al.\ (2014). 

A first step to introduce a functional heteroskedastic framework in the tradition of Engle (1982) was undertaken in H\"{o}rmann et al. (2013). These authors found conditions for the existence of functional ARCH(1) processes, for which the conditional variance depends on the whole (intra-day) path of the previous observation. In the spirit of Bollerslev (1986), this paper introduces a functional model in which the conditional volatility function of the present observation is given as a functional linear combination of the (intra-day) paths of both the past squared observation and its volatility function. The following definition is central.

\begin{definition} A sequence of random functions $(y_i\colon i\in\mathbb{Z})$ is called a functional GARCH process of orders (1,1), abbreviated by fGARCH(1,1), if it satisfies the equations
\begin{align}
\label{def-1}
y_i&=\sigma_i\vare_i, \\
\label{def-2}
\sigma_i^2&=\delta+\alpha y_{i-1}^2+\beta\sigma_{i-1}^2,
\end{align}
where $\delta$ is a nonnegative function and the integral operators $\alpha$ and $\beta$ map nonnegative functions to nonnegative functions, which is in short notation referred to as
\beq\label{cond-1}
\delta\geq 0,\quad \alpha\geq 0,\quad\beta\geq 0.
\eeq
Further, the innovations 
 \begin{align}\label{cond-2}
(\vare_i\colon i\in\mathbb{Z})
\quad\mbox{are independent and identically distributed random functions.}
\end{align}
\end{definition}
\noindent Additional assumptions will be introduced were needed. Note that the integral operators $\alpha$ and $\beta$ are, for $t\in[0,1]$, defined by $(\alpha x)(t)=\int \alpha(t,s)x(s)ds$ and $(\beta x)(t)=\int \beta(t,s)x(s)ds$, where $x$ is an arbitrary element of $L^2[0,1]$ and $\int$ is used to mean $\int_0^1$. The integral kernel functions $\alpha(s,t)$ and $\beta(s,t)$ are elements of the Hilbert space $L^2[0,1]^2$ and consequently bounded. 

The model given by \eqref{def-1} and \eqref{def-2} exhibits two time variables. The first is labeled by the integer $i$ and will in the present context often refer to trading day $i$, even though other time units are possible. The second time variable is labeled by the real-valued $t$ which, without loss of generality, takes values in the unit interval $[0,1]$. With $t$, intra-day trading time is parameterized. From this set-up, it should become clear that a modeling of intra-day seasonality may be accommodated by the fGARCH model. Doing this in a functional framework, may lead to efficiency gains over a multivariate modeling approach if the underlying process is sufficiently smooth. Similar to the univariate case argued in Bollerslev (1986), \eqref{def-2} should lead to a more flexible modeling of the functional volatility lag structure compared to the fARCH(1) process of H\"ormann et al.\ (2013). The focus is here on the first-order (1,1) case. Even in the univariate case, this is the most widely used GARCH process, used as default in the financial industry because it tends to work well (see, for example, the overview article by Zivot, 2009). Similar arguments should also apply in the more complex functional world.

The paper proceeds as follows. Theoretical properties of the fGARCH(1,1) process defined via \eqref{def-1} and \eqref{def-2} are discussed in Section \ref{sec:structure}. This section contains the main contributions of this work, establishing conditions for the existence of the functional stochastic difference equations \eqref{def-1} and \eqref{def-2}. Since one may view the functions $y_i$ and $\sigma_i^2$ both as members of the space of square-integrable and continuous functions, two different theorems are stated to cover both cases. The results on the structure of the fGARCH(1,1) process are accompanied by additional results on the estimation of quantities appearing in \eqref{def-2}. These are given in Section \ref{sec:estimation}. Empirical aspects are highlighted in Section \ref{sec:empirical}. Proofs of the theoretical results are given in Section~\ref{sec:proofs}.

\section{Structure}
\label{sec:structure}

This section discusses properties of the fGARCH(1,1) process, in particular notions of unique, strictly stationary solutions to the defining equations \eqref{def-1} and \eqref{def-2}. Two different settings are of interest: solutions in $L^2=L^2[0,1]$, the space of square integrable functions on $[0,1]$, and in $\cC[0,1]$, the space of continuous functions on $[0,1]$. In both cases, solutions can be characterized through properties of the random integral kernel 
\[
\gamma_0(t,s)=\alpha(t,s)\vare_0^2(s)+\beta(t,s).
\]
By assumption, $\gamma_0(s,t)$ is an element of $L^2[0,1]^2$ with probability one and, following Riesz and Sz.--Nagy (1990, p.\ 148), gives rise to the random integral operator $\gamma_0$ defined through $(\gamma x)(t)=\int\gamma_0(s,t)x(s)ds$. Note that $\gamma_0\geq 0$ (it maps nonnegative functions to nonnegative functions) and denote its (random) Hilbert--Schmidt norm by
\[
\|\gamma_0\|_\mathcal{S}=\bigg(\intt\gamma_0^2(t,s)dtds\bigg)^{1/2}.
\]
In the $L^2$ case and for the fGARCH(1,1), the Hilbert--Schmidt norm $\|\gamma_0\|_\mathcal{S}$ will play the role of the usual top Lyapunov exponent when characterizing the existence of stationary solutions to random difference equations (see, for example, Boucherol and Picard, 1992). One may also observe the similarity to the univariate GARCH(1,1) case discussed in Nelson (1990). Additional information on properties of function spaces may be found in the monograph Bosq (2000).
 
\begin{theorem}\label{l-2-th} 
Assume that \eqref{cond-1} and \eqref{cond-2} hold, that $\delta\in L^2$, and that $\vare_0\in L^4$ with probability one.

(i) If
\beq\label{cond-3}
-\infty\leq \mathbb{E}[\log\|\gamma_0\|_\mathcal{S}]%\intt \left\{ \alpha(t,s)\vare_0^2(s)+\beta(t,s) \right\}^2dtds\bigg]<0,
<0,
\eeq
then \eqref{def-1} and \eqref{def-2} have a unique, strictly stationary and nonanticipative solution in $L^2$.

(ii) If there is $\nu>0$ such that
\beq\label{cond-4}
\mathbb{E}[\|\gamma_0\|_\mathcal{S}^\nu]<1,
%\bigg[\left(\intt \left( \alpha(t,s)\vare_0^2(s)+\beta(t,s) \right)^2dtds\right)^{\nu/2}\bigg]<1,
\eeq
then \eqref{def-1} and \eqref{def-2} have a unique, strictly stationary and nonanticipative solution in $L^2$ and $\mathbb{E}[\|\sigma_0^2\|_2^{\nu}]<\infty$.
\end{theorem}

The proof of Theorem \ref{l-2-th} is given in Section \ref{sec:proofs}. If $(y_i\colon i\in\mathbb{Z})$ is an fARCH(1) process, then $\gamma_0(s,t)=\alpha(s,t)\vare_0^2(s)$. The condition for the existence of a strictly stationary solution for the fARCH(1) equations in $L^2$ used in H\"ormann et al.\ (2013) can be rewritten as $\mathbb{E}[\|\gamma_0\|^\alpha_\mathcal{S}]<1$. It can then be seen readily from an application of Jensen's inequality that this condition implies \eqref{cond-3}.

\begin{cor}\label{rem-1} 
The following statements are consequences of Theorem \ref{l-2-th}. 

(i) There is a functional $g$ such that
\beq\label{ber-1}
\sigma_{i}^2=g(\vare_{i-1}, \vare_{i-2}, \ldots),\qquad i\in\mathbb{Z},
\eeq
and therefore $(\sigma_i^2\colon i\in\mathbb{Z})$ and $(y_i\colon i\in\mathbb{Z})$ are Bernoulli shifts. 

(ii) The representation \eqref{ber-1} implies the ergodicity of $(\sigma_i^2\colon i\in\mathbb{Z})$ and $(y_i\colon i\in\mathbb{Z})$. 

(iii) Both $(\sigma_i^2\colon i\in\mathbb{Z})$ and $(y_i\colon i\in\mathbb{Z})$ are weakly dependent processes.

(iv) For $\ell\in\mathbb{N}$ and $i,j\in\mathbb{Z}$, let $\vare_{i,j}^{(\ell)}$ be independent copies of $\vare_0$ and define
\[
\sigma_{i,\ell}^2
=g(\vare_{i-1},\ldots \vare_{i-\ell},\vare_{i-\ell-1,i-\ell}^{(i)},\vare_{i-\ell-2,i-\ell}^{(i)}, \ldots).
\]
If \eqref{cond-4} is satisfied, then there is a constant\/ $0<\rho<1$ such that
\beq\label{ber-2}
\mathbb{E}\left[\|\sigma_i^2-\sigma_{i,\ell}^2\|_2^\nu\right]=O(\rho^\ell).
\eeq
\end{cor}

Statement (iv) of the corollary says that, for every $\ell$, the sequence $(\sigma_i^2\colon i\in\mathbb{Z})$ can be approximated with the $\ell$--dependent sequence $(\sigma_{i,\ell}^2\colon i\in\mathbb{Z})$ under a geometric rate. Consequently, if \eqref{cond-4} holds with some $\nu\geq 2$, then the partial sum process $(N^{-1/2}\sum_{i=1}^{\lfloor Nx\rfloor}\{\sigma_i^2-\mathbb{E}[\sigma_i^2]\})$ satisfies the functional central limit theorem in the space of square integrable processes; see Berkes et al.\ (2012) and Jirak (2013) for recent theoretical contributions in this area. 

If smoother trajectories are required of the fGARCH(1,1) and its volatility process, one may view $(y_i\colon i\in\mathbb{Z})$ and $(\sigma_i^2\colon i\in\mathbb{Z})$ as functional sequences in $\cC[0,1]$. To give the analog of Theorem \ref{l-2-th} for continuous functions, let $\bar\gamma_0=\int\gamma_0(\cdot,s)ds$ and observe that $\bar\gamma_0\in\cC[0,1]$. The next theorem shows that the role of $\|\gamma_0\|_\mathcal{S}$ for square integrable functions is assumed by $\|\bar\gamma_0\|_\cC=\sup_t|\bar\gamma_0(t)|$ for continuous functions. 

\begin{theorem}\label{sup-th} 
Assume that \eqref{cond-1} and \eqref{cond-2} hold, that $\delta\in \cC[0,1]$, and that $\vare_0\in \cC[0,1]$ with probability one.

(i) If
\beq\label{cond-5}
-\infty\leq \mathbb{E}%\log\sup\int ( \alpha(t,s)\vare_0^2(s)+\beta(t,s) )ds
\left[\log\|\bar\gamma_0\|_\mathcal{C}\right]<0,
\eeq
then \eqref{def-1} and \eqref{def-2} have a unique, strictly stationary and nonanticipative solution in $\cC[0,1]$.

(ii) If there is $\nu>0$ such that
\beq\label{cond-6}
\mathbb{E}\left[%\left(\sup_{0\leq t \leq 1}\int \left( \alpha(t,s)\vare_0^2(s)+\beta(t,s) \right)ds\right)^{\nu}
\|\bar\gamma_0\|_\cC^\nu\right]<1,
\eeq
then \eqref{def-1} and \eqref{def-2} have a unique, strictly stationary nonanticipative solution in $\cC[0,1]$ and $\mathbb{E}[\|\sigma_0^2\|_\cC^{\nu}]<\infty$.
\end{theorem}

\begin{cor}
\label{rem-1/2}
As in Corollary \ref{rem-1}, the sequences $(y_i\colon i\in\mathbb{Z})$ and $(\sigma_i\colon i\in\mathbb{Z})$ are weakly dependent. Under condition \eqref{cond-6} there is a constant $0<\rho<1$ such that
\[
\mathbb{E}\left[\|\sigma_i^2-\sigma_{i,\ell}^2\|_\cC^\nu\right]=O(\rho^\ell),
\]
where $\sigma_{i,\ell}^2$ is defined in Corollary \ref{rem-1}.
\end{cor}

The proofs of Theorem \ref{sup-th} and Corollary \ref{rem-1/2} are given in Section \ref{sec:proofs}. The next section covers an estimation procedure for the fGARCH(1,1) process. Empirical examples are part of Section \ref{sec:empirical}.

%%%%%%%%%%%%%%%%%%%%%%%%%%%%%%%%%%%%%%%%%%%%%%%%%%%

\section{Estimation}
\label{sec:estimation}
\setcounter{equation}{0}

The goal of this section is to obtain an estimate of the volatility equation \eqref{def-2}, that is, of the function $\delta$ and the operators $\alpha$ and $\beta$ in \eqref{def-2}. Since the objects involved in the estimation procedure are infinite-dimensional, an $M$-dimensional class $\Phi_M=\{\varphi_1, \varp_2, \ldots, \varp_M\}$ of orthonormal functions on $[0,1]$ is introduced to represent $\delta$, $\alpha$ and $\beta$ in the following way. It is assumed that $\delta$ can be represented as a linear combination of the functions in $\Phi_M$, that is,
\beq\label{pare-0}
\delta=\sum_{m=1}^Md_m\varp_m.
\eeq
It is further assumed that the integral kernels $\alpha(s,t)$ and $\beta(s,t)$ are elements of the span of $\Phi_M\times\Phi_M$, so that
\beq\label{pare-1}
\alpha(t,s)=\sum_{m,m^\prime=1}^Ma_{m,m^\prime}\varp_m(t)\varp_{m^\prime}(s)
\qquad\mbox{and}\qquad
\beta(t,s)=\sum_{m,m^\prime=1}^Mb_{m,m^\prime}\varp_m(t)\varp_{m^\prime}(s).
\eeq
With these requirements in place, the problem of estimating $\delta$, $\alpha$ and $\beta$ from a functional sample $y_1, y_2, \ldots, y_n$ reduces to estimating the set of real-valued parameters $\{d_m, a_{m,m^\prime}, b_{m,m^{\prime}}\colon m, m^\prime=1,\ldots,M\}$. To this end, project first $y_1^2,\ldots,y_n^2$ and $\sigma_1^2,\ldots,\sigma_n^2$ onto $\Phi_M$ and define the $M$-dimensional vectors $\by_i^{(2)}=(y_{i,1}^{(2)},\ldots,y_{i,M}^{(2)})^\T$ and $\bs_i^{(2)}=(s_{i,1}^{(2)},\ldots,s_{i,M}^{(2)})^\T$ through their entries $y_{i,m}^{(2)}=\la y_i^2,\varp_m\ra$ and $s_{i,m}^{(2)}=\la\sigma_i^2,\varp_m\ra$, where $\la\cdot,\cdot\ra$ denotes the inner product in $L^2$ and ${}^\T$ stands for the transpose of vectors and matrices. Using \eqref{def-2} it follows that
\beq\label{rec-e-1}
\bs_i^{(2)}=\bd+\bA \by_{i-1}^{(2)}+\bB\bs_{i-1}^{(2)},
\eeq
where $\bd=(d_1, \ldots, d_M)^\T$ with $d_m=\la \delta, \varp_m\ra$ and $\bA$ and $\bB$ are $M\times M$ matrices whose $(m,m^\prime)$th entries are given by $a_{m,m^\prime}$ and $b_{m,m^\prime}$. With this multivariate formulation of the volatility equation, iterating \eqref{rec-e-1} yields
\beq\label{rec-e-2}
\bs_i^{(2)}=\sum_{\ell=1}^\infty \bB^{\ell-1}\bd+\sum_{\ell=1}^\infty\bB^{\ell-1}\bA \by_{i-\ell}^{(2)},
\eeq
where $\bB^0=\bI$, the $M\times M$ identity matrix. Note that \eqref{rec-e-2} provides a representation for the projections of the unobservable volatilities in terms of the observations. Define the generic equation
\beq\label{rec-e-2/3}
\tilde{\bs}_i^{(2)}=\tilde{\bs}_i^{(2)}(\ubtheta)=\sum_{\ell=1}^\infty \ubB^{\ell-1}\ubd+\sum_{\ell=1}^{\infty}\ubB^{\ell-1}\ubA \by_{i-\ell}^{(2)},
\eeq
where $\ubtheta=(\ubd^\T, \ubA^\T, \ubB^\T)^\T\in \bTheta\subset\mathbb{R}^{M+2M^2}$.  It is clear that $\bs_i^{(2)}=\tilde{\bs}_i^{(2)}(\btheta)$, where $\btheta=(\bd^\T, \bA^\T, \bB^\T)^\T$ is the true parameter of the fGARCH(1,1) process $(y_i\colon i\in\mathbb{Z})$. Since this can always be achieved by an appropriate scaling, it is assumed without loss of generality that
\beq\label{rec-e-3}
\mathbb{E}[\vare_i^2(t)]=1 
\qquad\mbox{for all}\;t \in[0,1].
\eeq
If $\mathcal{F}_i=\sigma(\vare_\ell\colon\ell\leq i)$ denotes the $\sigma$-algebra generated by the all innovations up to $i$, then \eqref{rec-e-3} implies that $\mathbb{E}[\by_i^{(2)}|\mathcal{F}_i]=\bs_{i}^{(2)}$ and hence the least square estimator of $\mathbf{\theta}$ is given as the smallest value of the criterion function
$\sum_{i=1}^n\{\by_i^{(2)}-\tilde{\bs}^{(2)}_{i}(\ubtheta)\}^\T\{\by_i^{(2)}-\tilde{\bs}^{(2)}_{i}(\ubtheta)\}$. Since only $y_{1},\ldots y_n$ are observed, the actual parameter estimator is based on a truncated version of \eqref{rec-e-2/3}, namely on
\[
\hat{\bs}_i^{(2)}=\hat{\bs}_i^{(2)}(\ubtheta)
=\sum_{\ell=1}^{i-1} \ubB^{\ell-1}\ubd+\sum_{\ell=1}^{i-1}\ubB^{\ell-1}\ubA \by_{i-\ell}^{(2)}.
\]
This leads to the least squares estimator
\[
\hat{\ubtheta}_n=\mbox{argmin}\bigg\{
S_n(\ubtheta)= \sum_{i=2}^n\{\by_i^{(2)}-\hat{\bs}^{(2)}_{i}(\ubtheta)\}^\T\{\by_i^{(2)}-\hat{\bs}^{(2)}_{i}(\ubtheta)\}
\colon\ubtheta\in \bTheta
\bigg\},
\]
with $\hat{\ubtheta}_n=(\hat{\bd}_n^\T, \hat{\bA}_n^\T, \hat{\bB}_n^\T)^\T$. The main result of this section is the strong consistency of $\hat{\ubtheta}_n$. Some guarantees to ensure the identifiability of $\btheta$ are needed and collected next. Let $\|\cdot\|$ denote the Euclidean norm of vectors and matrices.

\begin{assu}\label{assu1}
It is assumed that \\
(A1)
%\beq\label{condi-1}
$\by_1^{(2)}\;\;\mbox{is not measurable with respect to  } {\mathcal F}_0$; \\
%\eeq
(A2) 
%\beq\label{condi-2}
$\bA\mbox{ is nonsingular and } \|\bB\|<1$; \\
%\eeq
(A3) 
%\beq\label{condi-3}
$\bTheta \mbox{ is a compact set and }\btheta \mbox{ is in the interior of } \bTheta$;
%\eeq
and \\
(A4) 
%\begin{align}\label{condi-4}
%
$\mbox{there are }0<c_1\mbox{ and }c_2<1 \mbox{ such that } c_1\leq|\mathrm{det}(\ubA)|\mbox{ and }\|\ubB\|\leq c_2
\mbox{ for all }\ubtheta=(\ubd^\T, \ubA^\T, \ubB^\T)^\T\in \bTheta$.%\notag
%\end{align}
\end{assu}

Part (A1) of Assumption \ref{assu1} means hat $\by_1^{(2)}$ cannot be predicted almost surely from its past. Parts (A3) and (A4) on the parameter space $\bTheta$ contain standard regularity conditions. 

\begin{theorem}\label{cons-th} 
Assume that \eqref{def-1}--\eqref{cond-2} and \eqref{cond-4} hold with $\nu=1$, that $\mathbb{E}[\|\vare_0\|_2^2]<\infty$ and \eqref{rec-e-3} is satisfied. Then $\hat{\ubtheta}_n$ is strongly consistent for $\btheta$ under Assumption \ref{assu1}, that is, 
\[
\hat{\ubtheta}_n\to\btheta
\qquad\mbox{a.s.}
\]
as $n\to \infty$.
\end{theorem}

To establish the asymptotic normality of $\hat{\ubtheta}_n$, introduce the quantities
\[
{\bH}_0
=\mathbb{E}\bigg[\frac{\partial\tilde{\bs}^{(2)}_0(\btheta)}{\partial\ubtheta}\bigg],
\qquad 
{\bJ}_0=
\mathbb{E}\Big[\big\{\by^{(2)}-\bs^{(2)}_0\big\}\big\{\by^{(2)}-\bs^{(2)}_0\big\}^\T\Big]
\]
and
\[
{\bQ}_0=
\mathbb{E}\bigg[ \bigg(\frac{\partial\tilde{\bs}^{(2)}_0(\btheta)}{\partial\ubtheta}  \bigg)^\T 
\bigg(\frac{\partial\tilde{\bs}^{(2)}_0(\btheta)}{\partial\ubtheta}\bigg)  \bigg].
\]
The matrices ${\bH}_0$ and ${\bQ}_0$ exist under the assumptions of Theorem \ref{cons-th}. If $\mathbb{E}[\|\sigma^2_0\|_2^4]<\infty$ and $\mathbb{E}[\|\vare^2_0\|_2^4]<\infty$, then ${\bJ}_0$ is well defined. Following the literature on the asymptotic normality of $M$-estimators, it is further assumed that
\beq\label{non-sin}
{\bQ}_0\;\;\mbox{is a non-singular matrix.}
\eeq

\begin{theorem}\label{norm-th}
Assume that the conditions of Theorem \ref{cons-th} are satisfied, \eqref{non-sin} holds and $\mathbb{E}[\|y_0^2\|_2^4]<\infty$. Then, as $n\to \infty$,
\[
\sqrt{n}(\hat{\ubtheta}_n-\btheta)
\stackrel{\mathcal D}{\to}\bN,
\]
where $\bN$ is a $2M^2+M$-dimensional normal vector with $\mathbb{E}[\bN]=\bzero$ and
$
\mathbb{E}[\bN\bN^\T]=\bQ^{-1}_0\bH^\T_0\bJ_0\bH_0\bQ^{-1}_0.
$

\end{theorem}

%\begin{rem}\label{rem-2/1}{\rm 
The proofs of Theorem \ref{cons-th} and \ref{norm-th} are given in Section \ref{sec:proofs}. Typical choices for the orthonormal class of functions $\Phi_M$ include the common Fourier and  $B$-spline bases, but other choices such as wavelets may be entertained as well. If the finite-dimensionality condition \eqref{pare-1} is not satisfied, then $\delta$, $\alpha$ and $\beta$ can typically still be well approximated by finite-dimensional functions assuming that $M$ is sufficiently large. While for theoretical considerations $M$ should grow with the sample size, it is worthwhile noting that in practice often a small choice of $M$ (less than, say, $5$) will already work reasonably well. This aspect was investigated in the simulations reported in Aue et al.\ (2015) and is also reported elsewhere in the functional data literature (see Horv\'ath and Kokoszka, 2012). %Estimating the coefficients of the finite-dimensional approximations, Theorem \ref{cons-th} can be used to get statistical inference for the functions defining $y_i$ in \eqref{def-1} and their volatilities $\sigma_i^2$ in \eqref{def-2}.

A different route for utilizing the estimation procedure that leads to a small class $\Phi_M$ is to follow the idea of functional principal components analysis (see Horv\'ath and Kokoszka, 2012). For example, it can be proposed to use the eigenfunctions of the covariance kernel
\[
D(t,s)=\mathrm{Cov}(y_0^2(t), y_0^2(s)),
\qquad
t,s\in[0,1],
\]
as a basis, so that a few large eigenvalues and their corresponding eigenfunctions capture the most important directions of randomness in $y_1^2, \ldots, y_n^2$ and hence in the unobservable volatilities $\sigma_1^2,\ldots,\sigma_n^2$. Since $D(t,s)$ is unknown, the estimation is based on the empirical covariance kernel
\beq
\label{dhat}
\hat{D}_n(t,s)=\frac{1}{n}\sum_{i=1}^n\big\{y_i^2(t)-\bar{y}^{(2)}_n(t)\big\}\big\{y_i^2(s)-\bar{y}^{(2)}_n(s)\big\},
\eeq
where $\bar{y}^{(2)}_n=\frac 1n\sum_{i=1}^ny_i^2$. The spectral decomposition of $\hat{D}_n(t,s)$ gives rise to the ordered empirical eigenvalues $\hat{\lambda}_1\geq \ldots\geq  \hat{\lambda}_n$ and the corresponding empirical eigenfunctions $\hat{\varphi}_1, \ldots, \hat{\varphi}_n$. Therefore, $\Phi_M$ may be replaced with the class of estimated eigenfunctions $\hat{\Phi}_M=\{\hat\varphi_1,\ldots,\hat\varphi_M\}$, with appropriately chosen $M$. The following result shows that consistency can be established also for the functional principal components approach to estimating the fGARCH(1,1) process.

\begin{cor}\label{rem-2/1}
If $\int\mathbb{E}[y_0^4(t)]dt<\infty$, then the result of Theorem \ref{cons-th} is retained if the random class $\hat{\Phi}_M$ is used in place of a deterministic class $\Phi_M$. 
\end{cor}
The proof of Corollary \ref{rem-2/1} is given in Section \ref{sec:proofs}.

%%%%%%%%%%%%%%%%%%%%%%%%%%%%%%%%%%%%%%%%%%%%%%%%%%%

\section{Empirical results}
\label{sec:empirical}
\setcounter{equation}{0}

%%%%%%%%%%%%%%%%%%%%%%%%%%%%%%%%%%%%%%%%%%%%%%%%%%%

\subsection{Simulations}
\label{subsec:sim}

\begin{figure}[ht]
        \centering
        \caption{Five consecutive fGARCH(1,1) observations (read row-wise starting from the top-left) generated from model \eqref{sim-2}--\eqref{sim-3}.}
        \label{sim-data}
        \begin{minipage}{.5\textwidth}
            \centering
            \includegraphics[width=.999\linewidth]{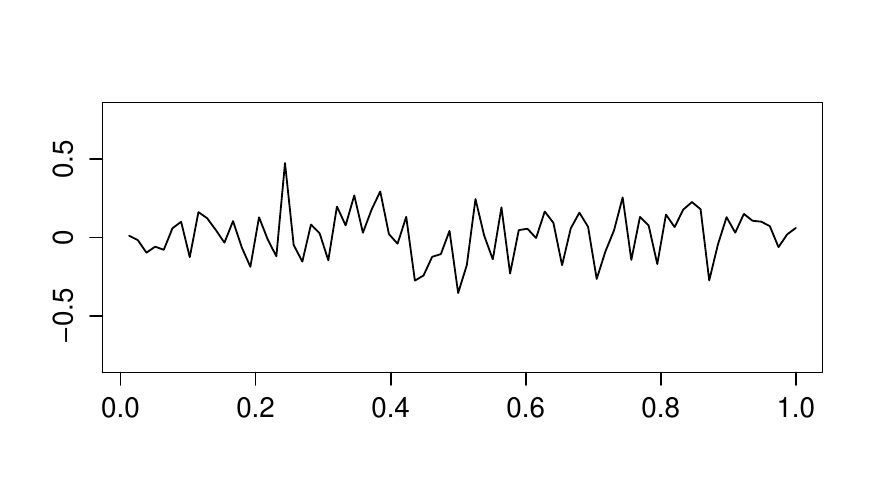}
        \end{minipage}%
       \begin{minipage}{.5\textwidth}
           \centering
           \includegraphics[width=.999\linewidth]{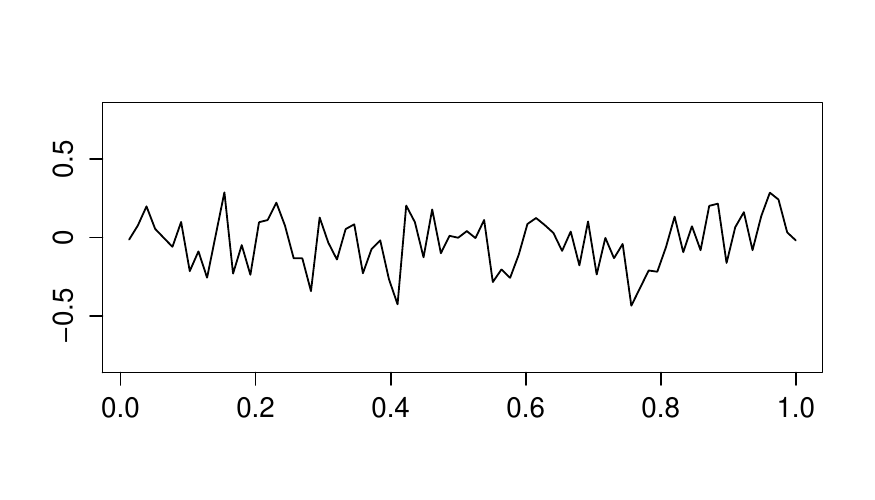}
       \end{minipage}
 \begin{minipage}{.5\textwidth}
            \centering
            \includegraphics[width=.999\linewidth]{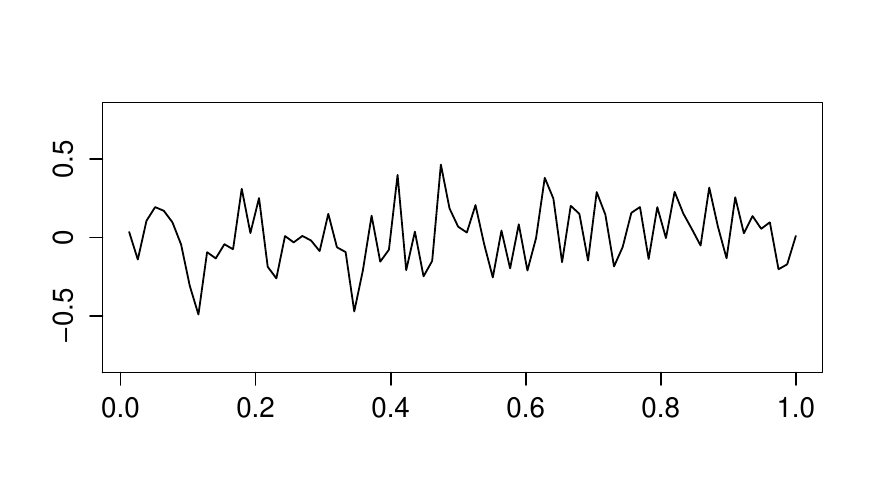}
        \end{minipage}%
       \begin{minipage}{.5\textwidth}
           \centering
           \includegraphics[width=.999\linewidth]{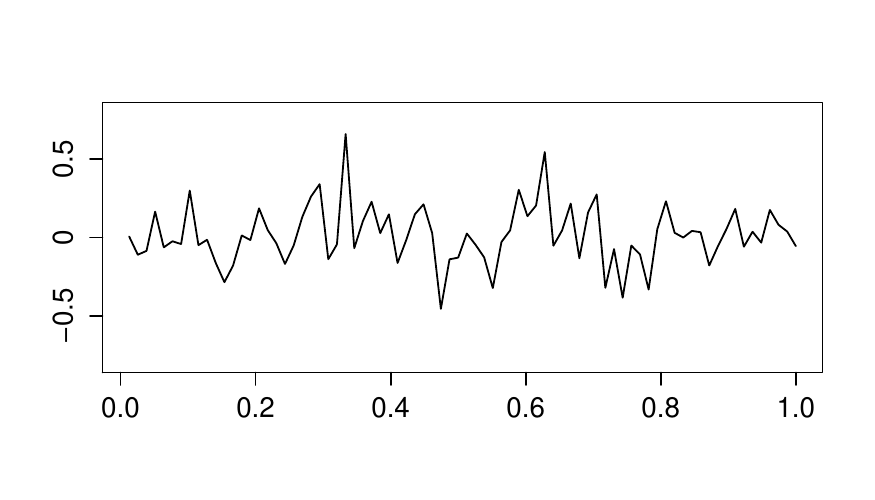}
       \end{minipage}
    \begin{minipage}{.5\textwidth}
           \centering
           \includegraphics[width=.999\linewidth]{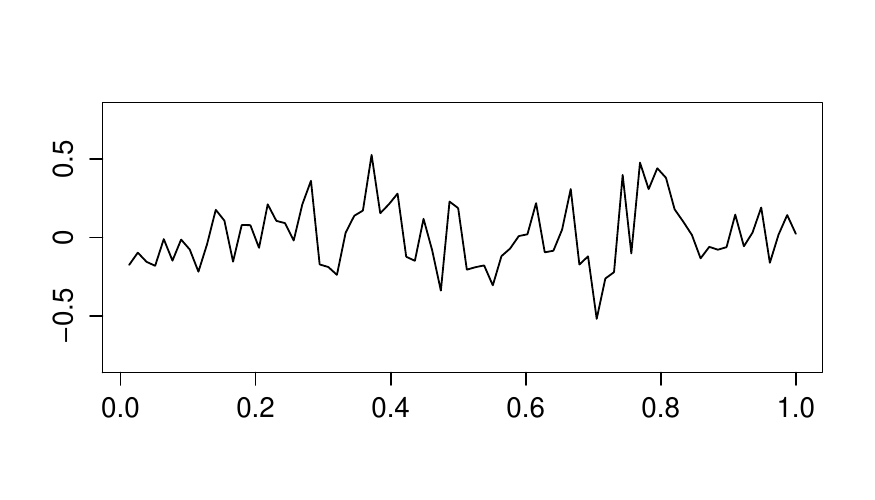}
        \end{minipage}
\end{figure}

In order to investigate the performance of the least squares procedure in finite samples, functional data was generated from the fGARCH(1,1) process $(y_i\colon i\in\mathbb{Z})$ whose volatility function \eqref{def-2} is specified through the constant intercept function $\delta$ given by 
\beq\label{sim-2}
\delta(t)=0.01,
\qquad t \in [0,1],
\eeq
and the operators $\alpha$ and $\beta$ defined through their kernel functions
\beq\label{sim-1}
\alpha(t,s)=\beta(t,s)=12t(1-t)s(1-s),
\qquad t,s\in [0,1].
\eeq
The innovations $(\vare_i\colon i\in\mathbb{Z})$ appearing in \eqref{def-1} and \eqref{def-2} are defined as
\beq\label{sim-3}
\vare_i(t)=\frac{\sqrt{\log 2}}{2^{200t}}B_i\bigg(\frac{2^{400t}}{\log 2}\bigg),
\qquad t \in[0,1],
\eeq
where $(B_i\colon i \in\mathbb{Z})$ are independent and identically distributed standard Brownian motions. These were constructed from the equi-spaced grid $j/285$, $j=1,\ldots,285$. Lengthy but elementary calculations show that \eqref{cond-3} holds in case of \eqref{sim-2}--\eqref{sim-3} and therefore $(y_i\colon i\in\mathbb{Z})$ constitutes a strictly stationary process. The recursion in \eqref{def-2} was initialized with  $\sigma^2_1=\delta$ and the first 1000 simulated curves were discarded as burn-in values to get close to the stationary solution. Figure \ref{sim-data} displays the first five simulated functions after the burn-in.

Random samples of size $n=300$, $600$ and $1200$ were obtained from the data generating process in equations \eqref{sim-2}--\eqref{sim-3}. Note that in this setting $\alpha(t,s)=\beta(t,s)=0.4\!\cdot\!\varphi_1(t)\varphi_1(s)$ with the normalized function $\varphi_1$ given by
\[
\varphi_1(t)=\sqrt{30}t(1-t),
\qquad t\in[0,1].
\] 
Equation \eqref{pare-1} shows then that the operator-related parameters to estimate are $a_{11}=b_{11}=0.4$. Note that the intercept function $\delta\equiv 0.01$ is not in the span of $\varphi_1$ and therefore does not satisfy \eqref{pare-0}, so that the corresponding least squares estimate becomes a proxy for $d_1=\la \delta, \varphi_1\ra=\sqrt{30}/600\approx .009$. The estimation procedure was run $1000$ times with $M=1$ yielding the outcomes reported in Table \ref{tab-1}. 

An alternative method for estimating $\delta$ can be obtained as follows. Observing that if $\hat{\alpha}$ and $\hat{\beta}$ are asymptotically consistent estimators for $\alpha$ and $\beta$ in the $L^2$-sense, then
\[
\tilde{\delta}_n=\bar{y}^{(2)}_n-(\hat\alpha+\hat\beta)\bar{y}^{(2)}_n
%\int(\hat{\alpha}_n(t,s)-\beta_n(t,s))\bar{y}^{(2)}_n(s)ds
\]
is an asymptotically consistent estimator for $\delta$, where $\bar{y}^{(2)}_n=\frac{1}{n}\sum_{i=1}^ny_i^2$. In the present simulation setting, there was little difference between $\hat{\delta}=\hat{d}_1\varphi$ and $\tilde{\delta}$, so only results for the first choice are presented here.

\begin{table}[htp]
\label{tab-1}
\caption{Estimates $\hat d_1$, $\hat a_{11}$ and $\hat b_{11}$ with $M=1$ and $\varphi_1(t)=\sqrt{30}t(1-t)$ in model \eqref{sim-2}--\eqref{sim-3} for different sample sizes $n$, with sample standard deviations given in brackets. The row labeled by $\infty$ shows the population values.}
\vspace{.2cm}
\centering
\begin{tabular}{clll}
\multicolumn{1}{c}{$n$} & \multicolumn{1}{c}{$\hat{d}_1$} & \multicolumn{1}{c}{$\hat{a}_{11}$} & \multicolumn{1}{c}{$\hat{b}_{11}$}\tabularnewline
\hline
300 & 0.013 (0.003) & 0.420 (0.058) & 0.306 (0.086) \tabularnewline
%\hline
600 & 0.011 (0.002) & 0.412 (0.042) & 0.344 (0.064) \tabularnewline
%\hline
1200 & 0.010 (0.001) & 0.408 (0.028) & 0.369 (0.045) \tabularnewline
\hline
$\infty$ & 0.009 & 0.400 & 0.400 \tabularnewline
%\hline
\end{tabular}
\end{table}

\begin{figure}[h!]
        \centering
        \caption{Plots of $\alpha(t,s)=\beta(t,s)$ of \eqref{sim-1} (top panel) and corresponding estimates $\hat{\alpha}_{1000}(t,s)$ and $\hat{\beta}_{1000}(t,s)$ (lower panel) based on the data generating process \eqref{sim-2}--\eqref{sim-3}.}\label{sim-alpha/beta}
        \begin{minipage}{.5\textwidth}
            \centering
            \includegraphics[width=.999\linewidth]{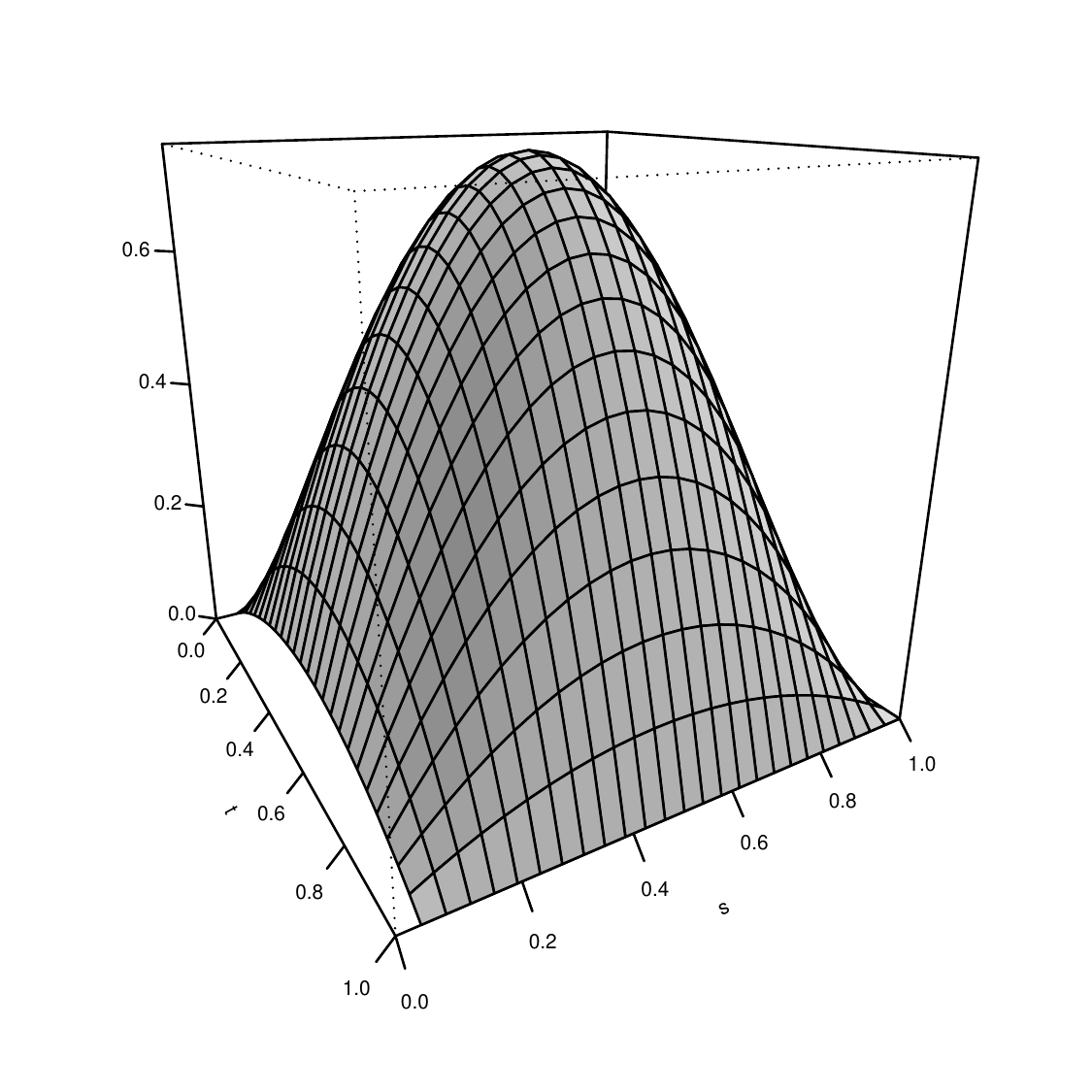}
        \end{minipage}\\%
%       \begin{minipage}{.5\textwidth}
%           \centering
%           \includegraphics[width=.999\linewidth]{delta_hat_from_simulated_data_20_basis.pdf}
%       \end{minipage}
 \begin{minipage}{.5\textwidth}
            \centering
            \includegraphics[width=.999\linewidth]{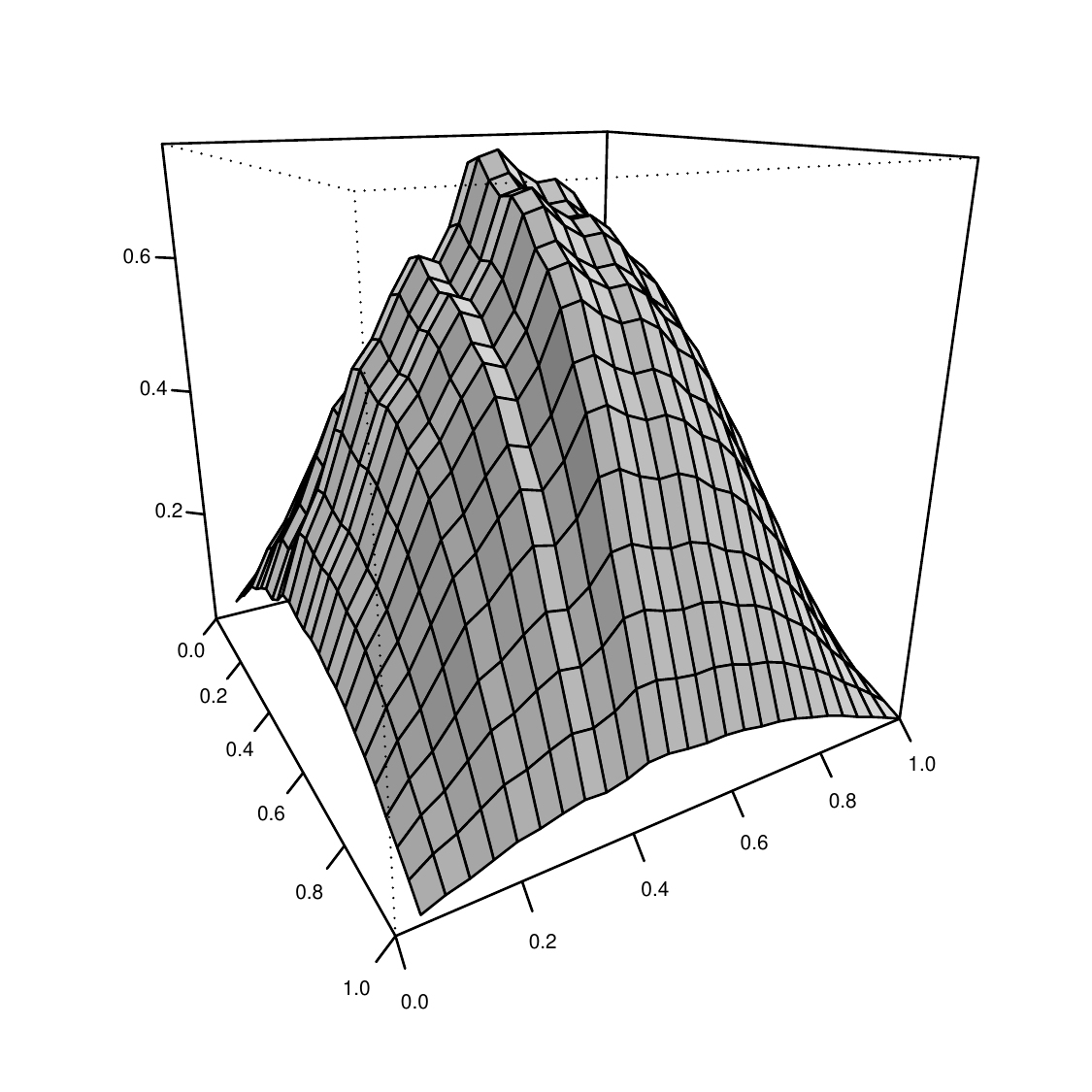}
        \end{minipage}%
       \begin{minipage}{.5\textwidth}
           \centering
           \includegraphics[width=.999\linewidth]{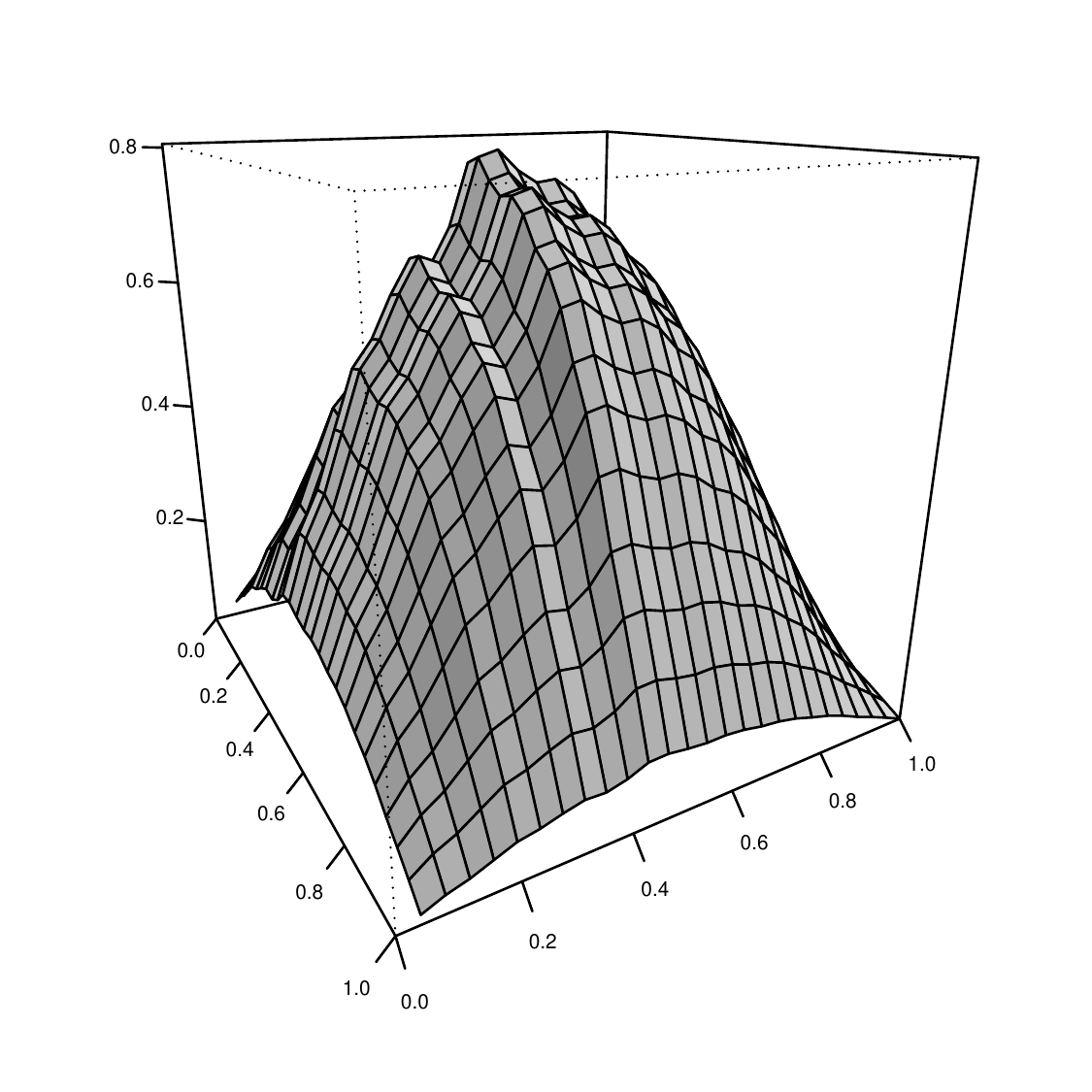}
       \end{minipage}
\end{figure}

To illustrate the estimation method using estimated functional principal components, whose large-sample behavior is given in Corollary \ref{rem-2/1}, $n=1000$ functional observations were generated from the fGARCH(1,1) model \eqref{sim-2}--\eqref{sim-3}. The largest eigenvalue of $\hat{D}_{1,000}$, the empirical covariance kernel defined in \eqref{dhat}, explains about 70\% of the variability in the squared functions $y_i^2$, while the second-largest eigenvalue only accounts for 4\%. First, the estimation procedure was applied using only the first empirical eigenfunction $\hat{\varphi}_1$. The resulting estimated kernels $\hat\alpha_{1000}(t,s)$ and $\hat\beta_{1000}(t,s)$ are shown in Figure \ref{sim-alpha/beta} along with the population kernels $\alpha(t,s)=\beta(t,s)$ given in \eqref{sim-3}. The experiment was then repeated using two or more eigenfunction but no qualitative improvement in the shape of the estimators $\alpha_{1000}(t,s)$ and $\beta_{1000}(t,s)$ was found. %As before, $\hat{d}_{1000}$ estimates the sum of the first coefficients of the expansion of $\delta$ in the empirical eigenbasis $\hat{\varphi}_1, \hat{\varphi}_2, \ldots$

%%%%%%%%%%%%%%%%%%%%%%%%%%%%%%%%%%%%%%%%%%%%%%%%%%%

\subsection{Application}
\label{subsec:app}

\begin{figure}[ht]
        \centering
        \caption{Plots of 5-minute log-returns of the SPY ETF from July 21 and 24--27, 2006 (read row-wise starting from the top-left).}\label{exch}
        \begin{minipage}{.5\textwidth}
            \centering
            \includegraphics[width=.999\linewidth]{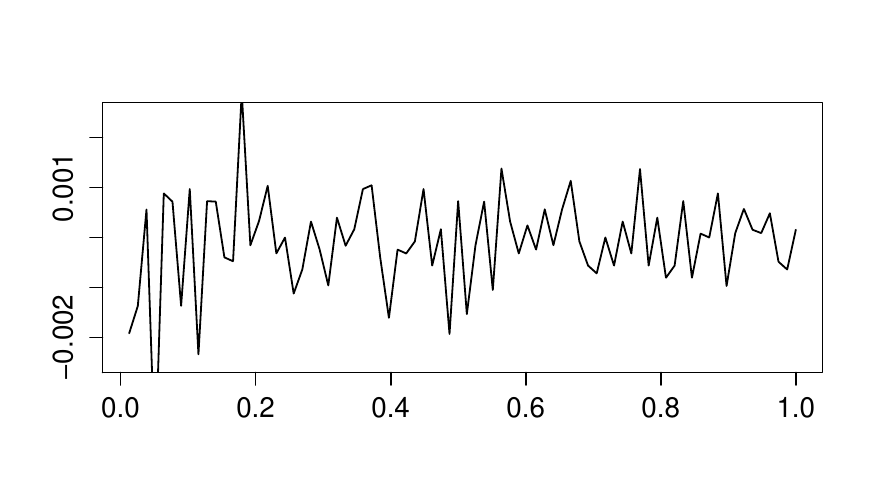}
        \end{minipage}%
       \begin{minipage}{.5\textwidth}
           \centering
           \includegraphics[width=.999\linewidth]{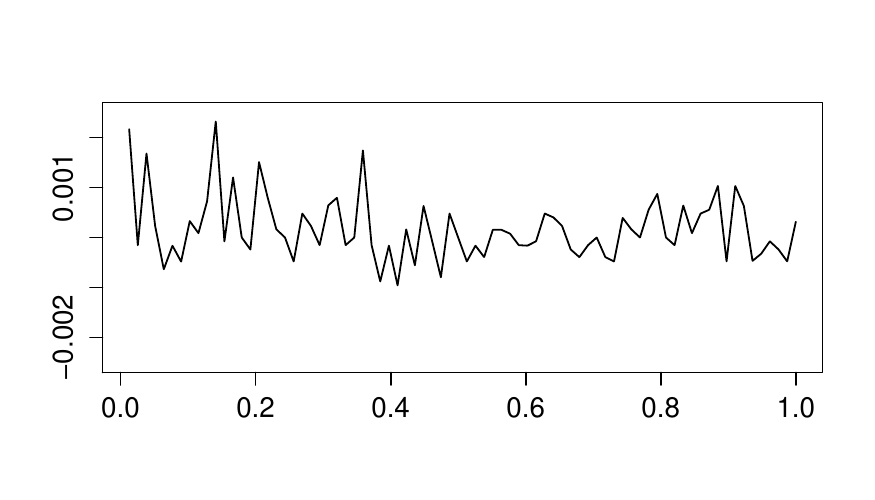}
       \end{minipage}
 \begin{minipage}{.5\textwidth}
            \centering
            \includegraphics[width=.999\linewidth]{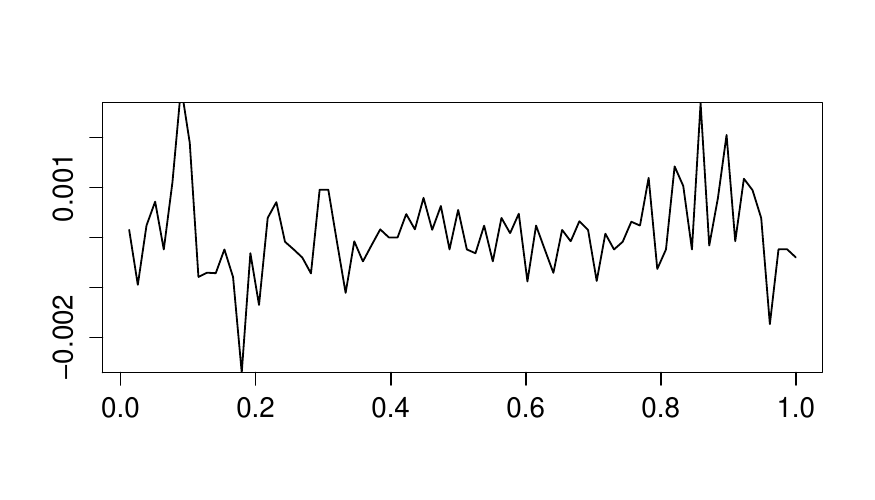}
        \end{minipage}%
       \begin{minipage}{.5\textwidth}
           \centering
           \includegraphics[width=.999\linewidth]{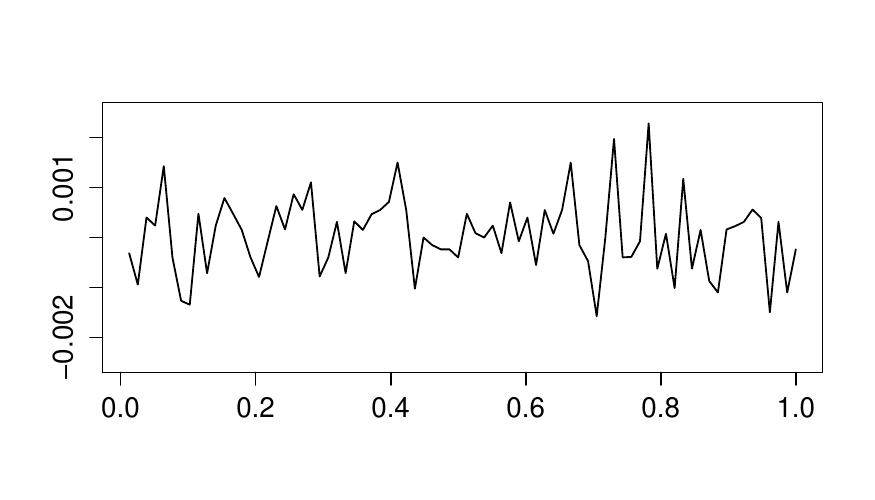}
       \end{minipage}
    \begin{minipage}{.5\textwidth}
           \centering
           \includegraphics[width=.999\linewidth]{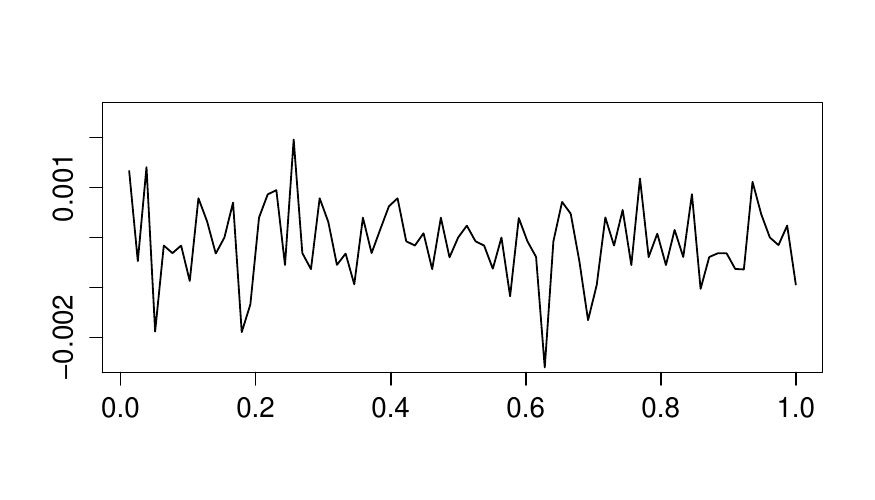}
        \end{minipage}
\end{figure}

The data example considered in this subsection is on intra-day log-returns of an exchange-traded fund (ETF), namely the popular SPDR S\&P 500 EDF (SPY) designed to track the S\&P 500 index. Let then $p_i(t)$ denote the value of the SPY ETF on day $i$ at time $t$. Following Cyree et al.\ (2004) one can set up the intra-day log-returns
\[
y_i(t)=\log p_i(t)-\log p_i(t-h),
\]
where $h$ is a time window typically corresponding to 1, 5, or 15 minutes. In the following, only $h=5$ and therefore 5-minute log-returns are considered. This is done for convenience because the price of the SPY ETF is recorded in five-minute resolution, for example, at the website {\tt https://quantquote.com}, and the 5-minute log-return process $(y_i)$ becomes easily computable. Since the SPY trades from 9:30am to 4:00pm, there are 78 measurements per day. The resulting high-dimensionality of the intra-day observations make the application of any multivariate ARCH or GARCH model impractical and a functional alternative is suggested in the following. Note that the volatility of the log-returns on day $i$ at time $t$ is then represented by 
\[
\sigma^2_i(t)=\mathrm{Var}(y_i(t)|{\mathcal F}_{i-1}),
\]
where ${\mathcal F}_{i-1}$ denotes the the available information until the end of day $i-1$. Figure \ref{exch} shows the graphs of the 5-minute log-returns of the SPY data for five consecutive trading days in July 2006.
\begin{figure}[ht]
        \centering
        \caption{Estimates for intercept function $\delta$ (top), and kernels $\alpha(t,s)$ (bottom left) and $\beta(t,s)$ (bottom right) based on the SPY ETF data.}\label{SPY-est}
        \begin{minipage}{.5\textwidth}
           \centering
           \includegraphics[width=.999\linewidth]{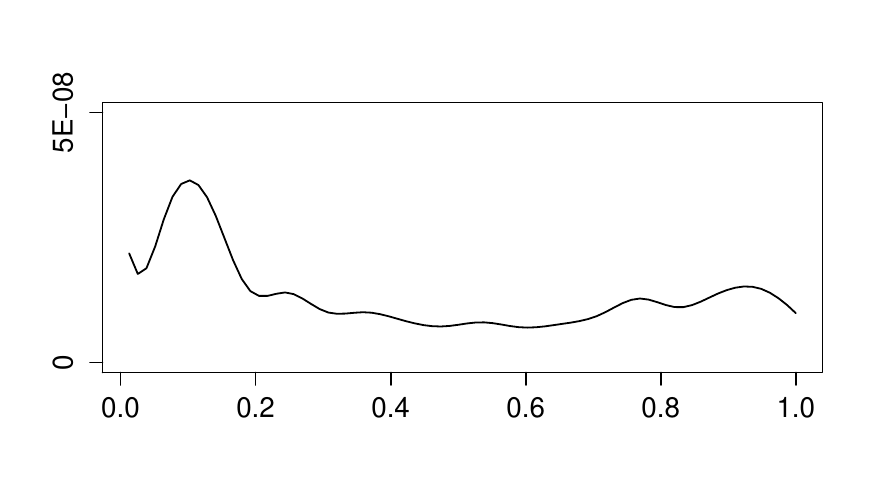}
        \end{minipage}
        \begin{minipage}{.5\textwidth}
            \centering
            \includegraphics[width=.999\linewidth]{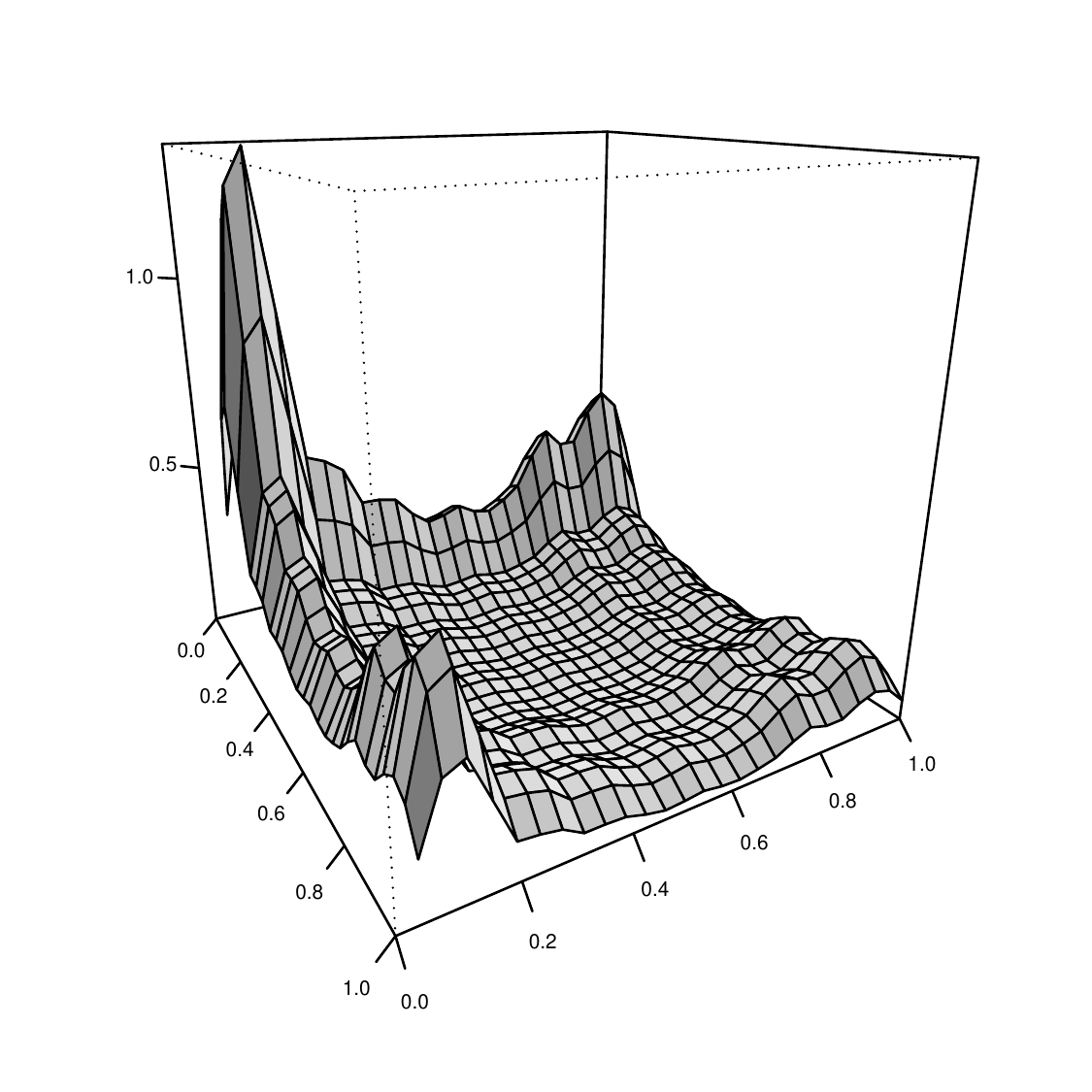}
        \end{minipage}%
       \begin{minipage}{.5\textwidth}
           \centering
           \includegraphics[width=.999\linewidth]{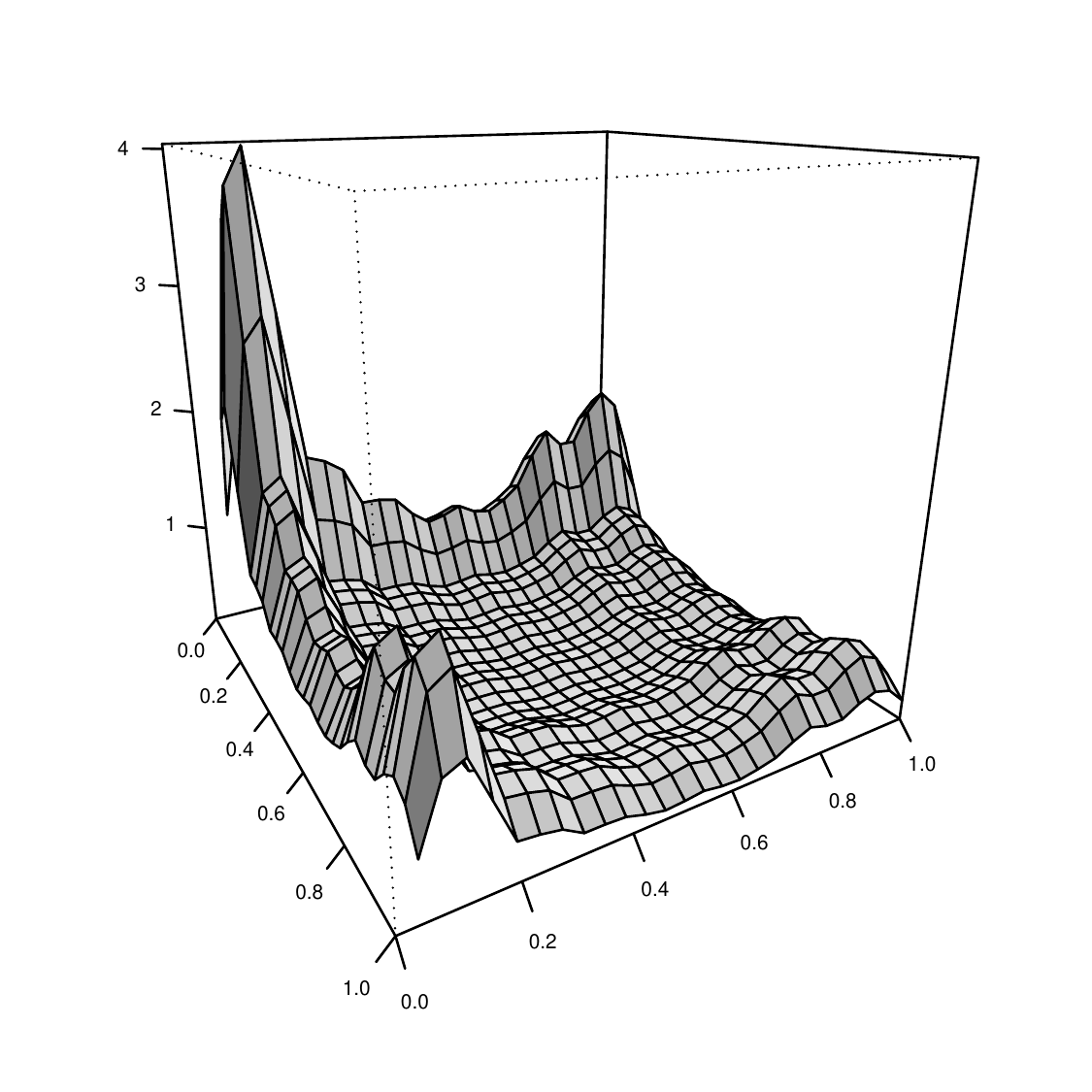}
       \end{minipage}
\end{figure}

The full data set considered here covers the trading days between January 2, 2003 and December 30, 2006. SPY is not traded on weekends and bank holidays. There were eight half trading days during this four year span and due to possibly different market behavior these eight curves were excluded from the analysis, resulting in the fitting of an fGARCH(1,1) process to $n=999$ functional observations. The least squares estimation procedure was carried out using $\Phi_M$ consisting of the empirical eigenfunctions associated with the $M$ largest eigenvalues of the covariance operator $\hat{D}_{999}$. Figure \ref{SPY-est} shows the estimates for the intercept function $\delta$ and integral kernels corresponding to the operators $\alpha$ and $\beta$ for the case $M=1$. Note that the largest eigenvalue explains 56\% of the variation. As in the simulated data example, increasing the number of eigenfunctions used, did not qualitatively improve the estimators. It should be noted that $\hat\delta$ has a global maximum shortly after trading opens in the morning and a further local maximum at the end of the day, corroborating similar findings in other papers.

\section{Proofs}
\label{sec:proofs}

\subsection{Proofs of Section \ref{sec:structure} results}

Using \eqref{def-1}, the volatility equation \eqref{def-2} can, for any $t\in[0,1]$, be written as
\begin{equation}\label{ga-def}
\sigma_i^2(t)=\delta(t)+\int \gamma_{i-1}(t,s)\sigma_{i-1}^2(s)ds,
\end{equation}
with integral kernel
\[
\gamma_{i-1}(t,s)=\alpha(t,s)\vare_{i-1}^2(s)+\beta(t,s).
\]
Iterating \eqref{ga-def} backward yields
 \begin{align*}
 \sigma_i^2(t)&=\delta(t)+\int \gamma_{i-1}(t,s_1)\left[\delta(s_1)+\int \gamma_{i-2}(s_1,s_2)\sigma^2_{i-2}(s_2)ds_2\right]ds_1\\
 &=\delta(t)+\int \gamma_{i-1}(t,s_1)\delta(s_1)ds_1+\intt \gamma_{i-1}(t,s_1)\gamma_{i-2}(s_1,s_2)\sigma_{i-2}^2(s_2)ds_2ds_1\\
 &=\delta(t)+\int \gamma_{i-1}(t,s_1)\delta(s_1)ds_1+\intt \gamma_{i-1}(t,s_1)\gamma_{i-2}(s_1,s_2)\delta(s_2)ds_2ds_1\\
 &\phantom{=\delta(t)}+\inttt \gamma_{i-1}(t,s_1)\gamma_{i-2}(s_1,s_2)\gamma_{i-3}(s_2,s_3)\sigma_{i-3}^2(s_3)ds_3ds_2ds_1,
 \end{align*}
 and so on. For $i\in\mathbb{Z}$, let $\Gamma_{i,0}$ denote the identity operator and, for $i\in\mathbb{Z}$ and $k\in\mathbb{N}$, let the random integral operators $\Gamma_{i,k}$ be defined by the equation 
\begin{displaymath}
(\Gamma_{i,k}\,x)(t)=
\int\cdots\int\gamma_{i-1}(t,s_1)\gamma_{i-2}(s_1,s_2)\cdots\gamma_{i-k}(s_{k-1},s_k)x(s_k)ds_k\cdots ds_2ds_1.
\end{displaymath}
Then, it follows that, for any $m\in\mathbb{N}$, 
\beq\label{s-m-def}
\sigma_i^2
=\sum_{k=0}^{m-1}\Gamma_{i,k}\,\delta+\Gamma_{i,m}\,\sigma_{i-m}^2.
\eeq
It can be seen from the definition of $\Gamma_{i,k}$ that
\begin{align}\label{ope}
\Gamma_{i,k}=\Gamma_{i,1}\circ\Gamma_{i-1,1}\circ\cdots\circ\Gamma_{i-k+1,1},
\end{align}
where $\circ$ denotes the composition of operators. The foregoing identifies 
\beq\label{s-def}
\sigma_i^2=\sum_{k=0}^\infty\Gamma_{i,k}\,\delta,
\eeq
as candidate for the nonanticipative solution of \eqref{def-1} and \eqref{def-2}. It remains to be shown that the sum on the right-hand side of equation \eqref{s-def} is well defined with probability one. This will be verified in the following.

%%%%%%%%%%%%%%%%%

\begin{proof}[Proof of Theorem \ref{l-2-th}]
(i) Consider the convergence of the sum in \eqref{s-def} in $L^2$, the space of square integrable functions on $[0,1]$. Since
\[
(\Gamma_{i,1}\,x)(t)=\int \gamma_{i-1}(t,s)x(s)ds
\]
defines an integral operator in $L^2$, it follows from Riesz and Sz.--Nagy (1990, p.\ 148)  that
\beq\label{sz-1}
\|\Gamma_{i,1}\|_{\cN}\leq \|\gamma_{i-1}\|_2,
\eeq
where $\|\cdot\|_{\cN}$ denotes the operator norm on $L^2$. Using representation \eqref{ope} and inequality \eqref{sz-1} leads to
\beq\label{op-2}
\|\Gamma_{i,k}\|_{\cN}\leq \prod_{\ell=1}^{k}\|\gamma_{i-\ell}\|_2.
\eeq
(The usual arrangement $\prod_{i\in \emptyset}=1$ applies here and in the following.)  Since $\|\gamma_\ell\|_2$ are independent and identically distributed random variables, the strong law of large numbers implies that 
\beq\label{ss-1}
\frac{1}{k}\sum_{\ell=1}^k\log\|\gamma_{i-\ell}\|_2
\to
\mathbb{E}\left[\log\|\gamma_{0}\|_2\right]
\qquad\mbox{a.s.}
\eeq
Hence, in view of \eqref{cond-3}, 
\[
P\bigg\{ \sum_{k=0}^\infty\Gamma_{i,k}\,\delta\in L^2[0,1]\bigg\}=1.
\]
It is easy to see that that the infinite sum on the right side of \eqref{s-def} satisfies \eqref{ga-def} and hence \eqref{def-1} and \eqref{def-2}. To prove the uniqueness of the solution to \eqref{def-1} and \eqref{def-2} assume that $\bar{\sigma}_i^2$ is another such solution. Using the backward recursions as in \eqref{s-m-def} shows that, for all $m\in\mathbb{N}$,
\beq\label{uni-1}
\bar{\sigma}_i^2
=\sum_{k=0}^{m-1}\Gamma_{i,k}\,\delta+\Gamma_{i,m}\,\bar{\sigma}_{i-m}^2.
\eeq
Hence, for all $m\in\mathbb{N}$,
\begin{align}\label{uni-2}
\bigg\| \bar{\sigma}_i^2-\sum_{k=0}^\infty\Gamma_{i,k}\,\delta\bigg\|_2
\leq 
\bigg\|\sum_{k=m}^\infty\Gamma_{i,k}\,\delta\bigg\|_2
+\left\|\Gamma_{i,m}\,\bar{\sigma}_{i-m}^2\right\|_2.
\end{align}
Assumption \eqref{cond-3} implies for the first norm on the right-hand side of the latter equation that, as $m\to\infty$,
\[
\bigg\|\sum_{k=m}^\infty\Gamma_{i,k}\,\delta\bigg\|_2
\leq \sum_{k=m}^\infty\|\Gamma_{i,k}\|_{\cN}\|\delta\|_2\to 0
\qquad\mbox{a.s.}
\]
Similarly,
\[
\left\|\Gamma_{i,m}\bar{\sigma}_{i-m}^2\right\|_2
\leq \|\Gamma_{i,m}\|_{\cN}\|\bar{\sigma}_{i-m}^2\|_2\stackrel{P}{\to}0,
\]
which follows since $\|\Gamma_{i,m}\|_{\cN}\to 0$ a.s.\ and, due to the weak stationarity of $(\bar{\sigma}_{i}\colon i\in\mathbb{Z})$, the sequence $(\|\bar{\sigma}_{i-m}^2\|_2\colon i\in\mathbb{Z})$ is bounded in probability. Therefore
\[
P\bigg\{\bigg\| \bar{\sigma_i}- \sum_{k=0}^\infty\Gamma_{i,k}\,\delta \bigg\|_2=0\bigg\}=1,
\]
which proves the first part of the theorem.

(ii) Note that
\[
\|\sigma_0^2\|_2
\leq \sum_{k=0}^\infty\|\Gamma_{0,k}\,\delta\|_2
\leq \|\delta\|_2 \sum_{k=0}^\infty\|\Gamma_{0,k}\|_{\cN}
\leq \|\delta\|_2 \sum_{k=0}^\infty\prod_{\ell=1}^k\|\gamma_{0-\ell}\|_2.
\]
Consequently Minkowski's inequality (cf.\ Hardy et al., 1959, pp.\ 24--26) implies that, for $\nu\geq 1$,
\begin{align*}
\left(\mathbb{E}[\|\sigma_0^2\|_2^\nu]\right)^{1/\nu}
\leq \|\delta\|_2\sum_{k=0}^\infty\bigg(\mathbb{E}\bigg(\prod_{\ell=1}^k\|\gamma_{0-\ell}\|_2\bigg)^\nu\bigg)^{1/\nu}
\leq \|\delta\|_2 \sum_{k=0}^\infty \left(\mathbb{E}[\|\gamma_{0}\|_2^\nu]\right)^{k/\nu}
\end{align*}
and, for $0<\nu\leq 1$,
\begin{align*}
\mathbb{E}[\|\sigma_0^2\|_2^\nu]
\leq \|\delta\|_2^\nu \sum_{k=0}^\infty \left(\mathbb{E}[\|\gamma_{0}\|_2^\nu]\right)^k
\end{align*}
The second part of the theorem is proved.
\end{proof}

%%%%%%%%%%%%%%%%%

\begin{proof}[Proof of Corollary \ref{rem-1}]
(i) The Bernoulli shift representation in \eqref{ber-1} is an immediate consequence of \eqref{s-def}.

(ii) Ergodicity is implied by the work of Stout (1974).

(iii) The statement follows immediately from (iv).

(iv) For $j\in\mathbb{N}$, let
\begin{displaymath}
\gamma_{i,\ell,i-\ell-j}(t,s)=\alpha(t,s)\{\vare^{(i)}_{i,\ell,i-\ell-j}(s)\}^2+\beta(t,s)
\end{displaymath}
and, for $k\geq \ell$, let the operators $\Gamma_{i,\ell,k}$ be defined through
\begin{align*}
(\Gamma_{i,\ell,k}\,x)(t)=\int\cdots\int
&\gamma_{i-1}(t,s_1)\gamma_{i-2}(s_1,s_2)\cdots\gamma_{i-\ell-1}(s_{\ell-2},s_{ \ell-1})\\
&\times\gamma_{i,\ell, i-\ell}(s_{\ell-1},s_{ \ell})\cdots\gamma_{i,\ell, i-k}(s_{k-1},s_k)x(s_k)ds_kds_{k-1}\cdots ds_1.
\end{align*}
Setting 
\[
\sigma^2_{i,\ell}
=\sum_{k=0}^{\ell-1}\Gamma_{i,k}\,\delta+\sum_{k=\ell}^\infty\Gamma_{i,\ell,k}\,\delta,
\]
it can be seen that
\[
\mathbb{E}\left[\|\sigma^2_{i}-\sigma^2_{i,\ell}\|_2^\nu\right]
\leq 2^\nu\mathbb{E}\bigg[\bigg\| \sum_{k=\ell}^\infty\Gamma_{i,k}\,\delta\bigg\|_2^\nu\bigg]
+2^\nu\mathbb{E}\bigg[\bigg\|\sum_{k=\ell}^\infty\Gamma_{i,\ell,k}\,\delta \bigg\|_2^\nu\bigg].
\]
It follows from the proof of Theorem \ref{l-2-th} that
\[
E\bigg[\bigg\|\sum_{k=\ell}^\infty\Gamma_{i,k}\,\delta\bigg\|_2^\nu\bigg]
\leq\left\{
\begin{array}{l@{\quad}l}
 \displaystyle \|\delta\|_2^\nu\bigg(\sum_{k=\ell}^\infty \{\mathbb{E}[\|\gamma_0\|_2^\nu]\}^{k/\nu}\bigg)^\nu,&\mbox{if}~~\nu\geq 1, \\[.5cm]
 \displaystyle\|\delta\|_2^\nu\sum_{k=\ell}^\infty\{\mathbb{E}[\|\gamma_0\|_2^\nu]\}^{k}, &\mbox{if}~~0<\nu \leq  1.
 \end{array}
 \right.
\]
Observing that
\[
\mathbb{E}\bigg[\bigg\|\sum_{k=\ell}^\infty\Gamma_{i,k}\,\delta\bigg\|_2^\nu\bigg]
=\mathbb{E}\bigg[\bigg\|\sum_{k=\ell}^\infty\Gamma_{i,\ell,k}\,\delta\bigg\|_2^\nu\bigg]
\]
and that $E[\|\gamma_0\|_2^\nu]<1$ by assumption, the proof of \eqref{ber-2} is complete.
\end{proof}

%%%%%%%%%%%%%%%%%

\begin{proof}[Proof of Theorem \ref{sup-th}]
(i) Recall relations \eqref{s-m-def}--\eqref{s-def}. It follows from the definition of $\Gamma_{i,k}$ that
\[
\|\Gamma_{i,k}\,\delta\|_{\cC}
\leq \|\delta\|_{\cC}\prod_{\ell=1}^k\left\|\bar{\gamma}_{i-\ell}\right\|_{\cC}
\]
with $\bar{\gamma}_{j}(t)=\int\gamma_j(t,s)ds$. Hence
\[
\bigg\|\sum_{k=0}^\infty\Gamma_{i,k}\,\delta\bigg\|_\cC
\leq\|\delta\|_{\cC}\sum_{k=0}^\infty\prod_{\ell=1}^k\left\|\bar{\gamma}_{i-\ell}\right\|_{\cC}.
\]
The strong law of large numbers yields that
\[
\frac{1}{k}\sum_{\ell=1}^k\log \left\|\bar{\gamma}_{i-\ell}\right\|_{\cC}
\to\mathbb{E}\left[\log \left\|\bar{\gamma}_{0}\right\|_{\cC}\right]
\qquad\mbox{a.s.},
\]
implying immediately that, under assumption \eqref{cond-5}
 \[
 P\bigg\{\sum_{k=0}^\infty\Gamma_{i,k}\,\delta\in \cC[0,1]\bigg\}=1.
 \]
Thus the existence of a solution to \eqref{def-1} and \eqref{def-2} is proven.  Assume that $(\bar{\sigma}_i^2\colon i\in\mathbb{Z})$ is an other solution to \eqref{def-1} and \eqref{def-2} in $\cC[0,1]$. Since \eqref{uni-1} holds, if follows as in \eqref{uni-2} that, for all $m\in\mathbb{N}$,
\begin{align*}
\bigg\| \bar{\sigma_i}-\sum_{k=0}^\infty\Gamma_{i,k}\,\delta\bigg\|_\cC
\leq \bigg\|\sum_{k=m}^\infty\Gamma_{i,k}\,\delta\bigg\|_\cC
+\left\|\Gamma_{i,m}\,\bar{\sigma}_{i-m}^2\right\|_\cC.
\end{align*}
If  \eqref{cond-5} is  satisfied, then, as $m\to\infty$,
\[
\bigg\|\sum_{k=m}^\infty\Gamma_{i,k}\delta\,\bigg\|_\cC
\leq\|\delta\|_\cC\sum_{k=m}^\infty\|\Gamma_{i,k}\|_{\cC}\to 0
\qquad\mbox{a.s.},
\]
using that $(1/m)\log\|\Gamma_{i,m}\|_{\cC}\to\mathbb{E}[\log\|\bar{\gamma}_0\|_{\cC}]<0 $ with probability one.
By the assumed  weak stationarity of the function series $(\bar{\sigma}_{i}\colon i\in\mathbb{Z})$, the univariate sequence $(\|\bar{\sigma}_{i-m}^2\|_\cC\colon i\in\mathbb{Z})$ is bounded in probability and therefore
\[
\left\|\Gamma_{i,m}\,\bar{\sigma}_{i-m}^2\right\|_\cC
\leq \|\Gamma_{i,m}\|_{\cC}\|\bar{\sigma}_{i-m}^2\|_\cC
\stackrel{P}{\to} 0.
\]
Consequently,
\[
P\bigg\{\bigg\|\bar{\sigma_i}-\sum_{k=0}^\infty\Gamma_{i,k}\,\delta\bigg\|_\cC=0\bigg\}=1,
\]
and the uniqueness of the solution to \eqref{def-1} and \eqref{def-2} is established.

(ii) The proof is similar to the proof of part (ii) of Theorem \ref{l-2-th}. 
\end{proof}

%%%%%%%%%%%%%%%%%

\begin{proof}[Proof of Corollary \ref{rem-1/2}]
The proof follows along the line of arguments used to establish Remark \ref{rem-1} and the details are thus omitted.
\end{proof}

%%%%%%%%%%%%%%%%%%%%%%%%%%%%%%%%%%

\subsection{Proofs of Section \ref{sec:estimation} results}

\begin{proof}[Proof of Theorem \ref{cons-th}]
First observe that, by \eqref{def-1}, $\la y_i^2, \varp_m\ra=\la \sigma_i^2\vare_i^2, \varp_m\ra= \la \sigma_i^2, \varp_m\ra+ \la \sigma_i^2(\vare_i^2-1), \varp_m\ra$. Due to the independence of $\vare_i$ and $\sigma_i, y_{i-\ell},\ell\in\mathbb{N}$, \eqref{rec-e-3} and the stationarity of $(y_i\colon i\in\mathbb{Z})$ imply that
\begin{align}\label{est-pr-1}
\mathbb{E}[(\mathbf{y}_i^{(2)}-\tilde{\mathbf{s}}_i^{(2)})^{\T}(\by_i^{(2)}-\tilde{\mathbf{s}}_i^{(2)})]
=\sum_{m=1}^ME\la \sigma_0^2(\vare_0^2-1), \varp_m\ra^2+M(\ubtheta),
\end{align}
where
\[
M(\ubtheta)
=\mathbb{E}[\{\bs_0^{(2)}-\tilde{\bs}_0^{(2)}(\ubtheta)\}^\T\{\bs_0^{(2)}-\tilde{\bs}_0^{(2)}(\ubtheta)\}], 
\qquad\ubtheta\in \bTheta.
\]
Since $M(\ubtheta)\geq 0$, $M(\btheta)=0$ and the first term on the right-hand side of \eqref{est-pr-1} is constant as a function of $\ubtheta$, it suffices to show that
\beq\label{est-pr-2}
M(\ubtheta)>0
\qquad\mbox{if}\quad
\ubtheta\neq \btheta.
\eeq
Assume now that \eqref{est-pr-2} does not hold, so that there is $\btheta_*\neq \btheta$, $\btheta_*\in\bTheta$ such that $M(\btheta_*)=0$. But this means that $\bs_0^{(2)}=\tilde{\bs}_0^{(2)}(\btheta_*)$ with probability one. Using \eqref{rec-e-2} and \eqref{rec-e-2/3}, it can be concluded that
\beq\label{est-pr-3}
\sum_{\ell=1}^\infty \bB^{\ell-1}\bd+\sum_{\ell=1}^\infty\bB^{\ell-1}\bA \by_{-\ell}^{(2)}=
\sum_{\ell=1}^\infty \bB^{\ell-1}_*\bd_*+\sum_{\ell=1}^{\infty}\bB_*^{\ell-1}\bA_* \by_{-\ell}^{(2)},
\eeq
with $\btheta_*=(\bd_*,\bA_* ,\bB^{}_*)$. If $\ell^*$ is the smallest integer such that  $\bB^{\ell^*-1}\bA\neq \bB_*^{\ell^*-1}\bA_*$, then
\[
(\bB^{\ell^*-1}\bA- \bB_*^{\ell^*-1}\bA_*)\by_{-\ell^*}^{(2)}=\sum_{\ell=1}^\infty \bB^{\ell-1}\bd-
\sum_{\ell=1}^\infty \bB^{\ell-1}_*\bd_*+\sum_{\ell=\ell^*+1}^\infty(\bB^{\ell-1}\bA -\bB_*^{\ell-1}\bA_*) \by_{-\ell}^{(2)}
\]
and therefore $\by_{-\ell^*}^{(2)}$ is measurable with respect to ${\mathcal F}_{-\ell^*-1}$ which contradicts part (A1) of Assumption \ref{assu1} due to stationarity. Hence  $\bB^{\ell}\bA= \bB_*^{\ell}\bA_*$ for all $\ell\geq 0$ resulting in $\bA=\bA_*$ and consequently $\sum_{\ell=1}^\infty \bB^{\ell-1}\bA =\sum_{\ell=1}^\infty \bB^{\ell-1}_*\bA$, in turn implying that $(\bI-\bB)^{-1}\bA=(\bI-\bB_*)^{-1}\bA$. Thus $\bB=\bB_*$. Since $\bA=\bA_*$ and $\bB=\bB_*$, \eqref{est-pr-3} shows that $\sum_{\ell=1}^\infty \bB^{\ell-1}\bd=\sum_{\ell=1}^\infty \bB^{\ell-1}\bd_*$ and hence $(\bI-\bB)^{-1}\bd=(\bI-\bB)^{-1}\bd_*$, completing the proof of $\bd=\bd_*$. By contradiction, \eqref{est-pr-2} is established. 

The ergodic theorem (cf.\ Breiman, 1968) yields that, for all $\ubtheta\in\bTheta$,
\[
\tilde{Z}_n(\ubtheta)=
\frac{1}{n}\sum_{i=2}^n\big\{\by_i^{(2)}-\tilde{\bs}^{(2)}_{i}(\ubtheta)\big\}^\T\big\{\by_i^{(2)}-\tilde{\bs}^{(2)}_{i}(\ubtheta)\big\}
\to\sum_{m=1}^M\mathbb{E}\left[\la \sigma_0^2(\vare_0^2-1), \varp_m\ra^2\right]+M(\ubtheta)
\qquad\mbox{a.s.}
\]
as $n\to \infty$. Observing that the derivatives of $\tilde{\bs}^{(2)}_{0}$ have a uniformly bounded expected value on $\ubtheta\in \bTheta$, standard arguments (see, for example, the proof of the uniform law of large numbers in Ferguson, 1996) show that
%\beq\label{pf-0}
\[
\sup_{\ubtheta\in\bTheta}
\bigg|\tilde{Z}_n(\ubtheta)
%\frac{1}{n}\sum_{i=2}^n\{\by_i^{(2)}-\tilde{\bs}^{(2)}_{i-1}(\ubtheta)\}^\T\{\by_i^{(2)}-\tilde{\bs}^{(2)}_{i-1}(\ubtheta)\}
-\bigg(\sum_{m=1}^M\mathbb{E}\left[\la \sigma_0^2(\vare_0^2-1),\varp_m\ra^2\right]+M(\ubtheta)\bigg)\bigg|
\to 0\qquad\mbox{a.s.}
%\eeq
\]
Using \eqref{rec-e-2/3} and the assumption that $\|\ubB\|\leq \rho<1$ on $\bTheta$ gives
\begin{align}
\label{fast}
n\sup_{\ubtheta\in\bTheta} 
\big|\tilde{Z}_n(\ubtheta)-\hat{Z}_n(\ubtheta)\big|
= \mathcal{O}(1) \qquad a.s.,
%\biggl|\frac{1}{n}\sum_{i=2}^n\{\by_i^{(2)}-\tilde{\bs}^{(2)}_{i-1}(\ubtheta)\}^\T\{\by_i^{(2)}-\tilde{\bs}^{(2)}_{i-1}(\ubtheta)\}
%-\frac{1}{n}\sum_{i=2}^n\{\by_i^{(2)}-\hat{\bs}^{(2)}_{i-1}(\ubtheta)\}^\T\{\by_i^{(2)}-\hat{\bs}^{(2)}_{i-1}(\ubtheta)\}\biggl|
%\to 0\qquad\mbox{a.s.}
\end{align}
where $\hat{Z}_n(\ubtheta)$ is defined as $\tilde{Z}_n(\ubtheta)$ but using $\hat{\bs}^{(2)}_{i}(\ubtheta)$ in place of $\tilde{\bs}^{(2)}_{i}(\ubtheta)$. Thus,
\beq\label{pf-1}
\sup_{\ubtheta\in\bTheta}
\bigg|\hat{Z}_n(\ubtheta)
%\frac{1}{n}\sum_{i=2}^n\{\by_i^{(2)}-\hat{\bs}^{(2)}_{i-1}(\ubtheta)\}^\T\{\by_i^{(2)}-\hat{\bs}^{(2)}_{i-1}(\ubtheta)\}
-\bigg(\sum_{m=1}^M\mathbb{E}\left[\la \sigma_0^2(\vare_0^2-1),\varp_m\ra^2\right]+M(\ubtheta)\bigg)\bigg|
\to 0\qquad\mbox{a.s.}
\eeq
The results in \eqref{est-pr-2} and \eqref{pf-1} imply the consistency in Theorem \ref{cons-th} via standard arguments.
\end{proof}

\begin{proof}[Proof of Theorem \ref{norm-th}]
Note first that
\beq\label{def0}
\frac{\partial \hat{Z}_n(\hat{\ubtheta}_n)}{\partial\ubtheta}=\bzero.
\eeq
Reasoning as on page 466 of Seber and Lee (2003) yields
\[
\frac{\partial \tilde{Z}_n({\ubtheta})}{\partial\ubtheta}
=-\frac{2}{n}\sum_{i=2}^n
\bigg(\frac{\partial \tilde{\bs}^{(2)}_i(\ubtheta)}{\partial \ubtheta}\bigg)^\T
\big\{\by^{(2)}_i-\tilde{\bs}_i^{(2)}(\ubtheta)\big\}
\]
and
\[
\frac{\partial^2 \tilde{Z}_n({\ubtheta})}{\partial\ubtheta^2}
=\frac{2}{n}\sum_{i=2}^n
\bigg(\frac{\partial \tilde{\bs}^{(2)}_i(\ubtheta)}{\partial \ubtheta}\bigg)^\T
\bigg(\frac{\partial \tilde{\bs}^{(2)}_i(\ubtheta)}{\partial \ubtheta}\bigg)
+\frac{2}{n}\sum_{i=2}^n\bR_i(\ubtheta),
\]
where
\[
\bR_i(\ubtheta)=
\bigg\{\sum_{m=1}^M
\bigg(\frac{\partial^2 \tilde{s}^{(2)}_{i,m}(\ubtheta)}{\partial \utheta_j\partial \utheta_k}\bigg)
\big\{\la y_i^2,\varphi_m \ra-\tilde{s}^{(2)}_{i,m}(\ubtheta)\big\}
\colon j,k = 1,\ldots, 2M^2+M\bigg\},
\]
with $\ubtheta=(\utheta_1, \ldots, \utheta_{2M^2+M})^\T$ and $\tilde{\bs}_i^{(2)}(\ubtheta)=(\tilde{s}_{i,1}^{(2)}(\ubtheta),\ldots ,\tilde{s}_{i,M}^{(2)}(\ubtheta))^\T$.
Let
\[
{\bQ}(\ubtheta)
=\mathbb{E}\bigg[\bigg(\frac{\partial \tilde{\bs}^{(2)}_0(\ubtheta)}{\partial \ubtheta}\bigg)^\T
\bigg(\frac{\partial \tilde{\bs}^{(2)}_0(\ubtheta)}{\partial \ubtheta}\bigg)\bigg]
+\mathbb{E}[\bR_0(\ubtheta)].
\]
It is easy to see that
\begin{align}\label{uni-11}
\bQ(\ubtheta)\;\;\mbox{is coninuous on} \;\;\bTheta\;\;\mbox{and} \;\;\bQ_0=\bQ(\btheta).
\end{align}
Hence the ergodic theorem (see again the proof of the uniform law of large numbers in Ferguson, 1996) implies that
\begin{align}\label{uni-22}
\sup_{\ubtheta\in \bTheta} 
\bigg\| \frac{\partial^2 \tilde{Z}_n({\ubtheta})}{\partial\ubtheta^2} -2\bQ(\ubtheta)\bigg\|\to0
\qquad\mbox{a.s.}
\end{align}
The central limit theorem for orthogonal martingale differences in Billingsley (1968) yields moreover that
\begin{align}\label{uni-3}
\sqrt{n}\frac{1}{2}\frac{\partial \tilde{Z}_n({\btheta})}{\partial\ubtheta}
\stackrel{\mathcal D}{\to}\tilde{\bN},
\end{align}
where $\tilde{\bN}$ is a $2M^2+M$-dimensional normal random vector with $\mathbb{E}[\tilde{\bN}]=\bzero$ and
$\mathbb{E}[\tilde{\bN}\tilde{\bN}^\T]=\bH_0^\T\bJ_0\bH_0$.

Arguing along the lines of the justification of \eqref{fast} one can verify that
\beq\label{fast-1}
n\sup_{\ubtheta\in \bTheta}
\bigg\| \frac{\partial \tilde{Z}_n({\ubtheta})}{\partial\ubtheta}-\frac{\partial \hat{Z}_n({\ubtheta})}{\partial\ubtheta}\bigg\|
=\mathcal{O}(1)\qquad\mbox{a.s.}
\eeq
and
\beq\label{fast-2}
n\sup_{\ubtheta\in \bTheta}
\bigg\| \frac{\partial^2 \tilde{Z}_n({\ubtheta})}{\partial\ubtheta^2}-\frac{\partial^2 \hat{Z}_n({\ubtheta})}{\partial\ubtheta^2} \bigg\|
=\mathcal{O}(1)\qquad\mbox{a.s.}
\eeq
Combining \eqref{def0} with the coordinate-wise mean value theorem for vectors leads to
\begin{align}\label{mvt}
-\frac{\partial \hat{Z}_n({\btheta})}{\partial\ubtheta}=\hat{\bV}_n(\tilde{\ubtheta}_n-\btheta),
\end{align}
with some matrix $\hat{\bV}_n$. From \eqref{uni-11}, \eqref{uni-22} and \eqref{fast-2}, it can be concluded that the matrix $\hat{\bV}_n$ satisfies
\beq\label{conver-1}
\hat{\bV}_n\;\to\;2\bQ_0.
\eeq
Displays \eqref{uni-3} and \eqref{fast-1} also yield that
\begin{align}\label{uni-clt}
\sqrt{n}\frac{1}{2}\frac{\partial \hat{Z}_n({\btheta})}{\partial\ubtheta}
\stackrel{\mathcal D}{\to}\tilde{\bN},
\end{align}
where $\tilde{\bN}$ is defined in \eqref{uni-3}. Putting together \eqref{mvt}--\eqref{uni-clt}, the asymptotic normality of the estimator in Theorem \ref{norm-th} is established.
\end{proof}

\begin{proof}[Proof of Corollary \ref{rem-2/1}]
The approximability established in Corollary \ref{rem-1} implies that $\|\hat{D}_n-D\|\to 0$ in probability and therefore
\[
\max_{1\leq m\leq M}\|\hat{\varphi}_m-\hat{\zeta}_m\varphi_m\|\stackrel{P}{\to}0,
\qquad (n\to\infty),
\]
where the $\hat{\zeta}_m$'s are random signs, assuming that $\lambda_1>\cdots >\lambda_M>\lambda_{M+1}$ are the eigenvalues of $D$ in decreasing order (see Horv\'ath and Kokoszka, 2012). Hence the consistency of  Theorem \ref{cons-th} remains true when the functions in $\hat\Phi_M=\{\hat{\varphi}_1,\ldots,\hat{\varphi}_M\}$ are used to set up \eqref{pare-1}.
\end{proof}

%%%%%%%%%%%%%%%%%

\end{document}